\documentclass[a4paper,USenglish,cleveref, autoref, thm-restate,authorcolumns]{lipics-v2021}
\newcommand{\vlong}[1]{#1}
\newcommand{\vshort}[1]{}

\def\BaseUrl{https://dominik-kirst.github.io/mca/MCA.html\#}
\newcommand{\coqdoc}[1]{\href{\BaseUrl#1}{\raisebox{-.9mm}{\includegraphics[height=1em]{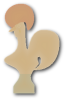}}}}

\pdfoutput=1 
\hideLIPIcs  
  
\title{From Partial to Monadic: Combinatory Algebra with Effects
}

\authorrunning{Cohen, Grunfeld, Kirst, Miquey}

\Copyright{Cohen, Grunfeld, Kirst and Miquey} 

\author{Liron Cohen}{Ben-Gurion University,  Israel}{cliron@bgu.ac.il}{0000-0002-6608-3000}{}
\author{Ariel Grunfeld}{Ben-Gurion University,  Israel}{arielgru@post.bgu.ac.il}{0009-0008-1142-282X}{}
\author{Dominik Kirst}{Ben-Gurion University,  Israel / Inria Paris, France}{dominik.kirst@inria.fr	}{0000-0003-4126-6975}{}
\author{\'{E}tienne Miquey}{Aix Marseille University, CNRS, I2M, Marseille, France}{etienne.miquey@univ-amu.fr	}{0000-0002-5987-6547}{}



\supplement{Rocq Formalization}

\supplementdetails[subcategory={Source Code}, cite={}, swhid={}]{Other}{https://github.com/dominik-kirst/mca}

\funding{This work was partially supported by Grant No. 2020145 from the United States-Israel Binational Science Foundation (BSF). Dominik Kirst received funding from the European Union’s Horizon research and innovation programme under the Marie Skłodowska-Curie grant agreement No.101152583 and a Minerva
Fellowship of the Minerva Stiftung Gesellschaft für die Forschung mbH. }

\acknowledgements{We are deeply grateful to Ross Tate for his valuable ideas and insights, which significantly shaped the direction of this work.} 

\nolinenumbers 

\EventEditors{Maribel Fern\'{a}ndez}
\EventNoEds{1}
\EventLongTitle{10th International Conference on Formal Structures for Computation and Deduction (FSCD 2025)}
\EventShortTitle{FSCD 2025}
\EventAcronym{FSCD}
\EventYear{2025}
\EventDate{July 14--20, 2025}
\EventLocation{Birmingham, UK}
\EventLogo{}
\SeriesVolume{337}
\ArticleNo{11}

\usepackage{cite}
\usepackage{amsmath,amssymb,amsfonts}
\usepackage{algorithmic}
\usepackage{graphicx}
\usepackage{textcomp}
\usepackage{xcolor}
\usepackage{makecell}

\usepackage{soul}
\usepackage{tikz-cd}
\usepackage{xspace}
\usepackage{xcolor}
\usepackage{perfectcut}

\usepackage[final,inline,nomargin]{fixme}
\fxusetheme{color}
\fxuseenvlayout{color}
\FXRegisterAuthor{lc}{alc}{\color{purple}[L]}
\FXRegisterAuthor{ag}{aag}{\color{teal}[A]}
\FXRegisterAuthor{dk}{adk}{\color{magenta}[D]}
\FXRegisterAuthor{em}{aem}{\color{blue}[E]}
\FXRegisterAuthor{rt}{art}{\color{orange}[R]}
\FXRegisterAuthor{rev}{arev}{\color{red}[Reviewer]}

\RequirePackage{xparse}

\usepackage[colorinlistoftodos,prependcaption,textsize=scriptsize,textwidth=2cm]{todonotes}
\usepackage[inline]{enumitem}
\usepackage{multicol}
\usepackage{mathpartir}
\usepackage{marvosym}
\usepackage{float}
\usepackage{mathbbol}
\usepackage{mathtools}
\usepackage{mathrsfs}
\usepackage{amsmath}
\usepackage{amsthm}
\usepackage{thmtools}
\usepackage[notext,nomath]{stix}
\usepackage{stmaryrd}
\usepackage{proof}
\usepackage{hyperref}

\usepackage{stix}
\usepackage{array}
\usepackage{arydshln}
\usepackage{upgreek}
\usepackage{empheq}
\usepackage{perfectcut}

\usepackage{flafter} 
\usetikzlibrary{matrix, positioning, fit, arrows, chains,shapes}

\usepackage{mathtools}
\usepackage[bb=boondox]{mathalfa}
\usepackage{proof}
\hypersetup{colorlinks=true}

\Crefname{equation}{Eq.}{Eqs.}
\Crefname{figure}{Fig.}{Figs.}
\Crefname{tabular}{Tab.}{Tabs.}
\Crefname{section}{Sec.}{Secs.}
\Crefname{definition}{Def.}{Defs.}
\Crefname{defi}{Def.}{Defs.}
\Crefname{lemma}{Lem.}{Lems.}
\Crefname{lem}{Lem.}{Lems.}
\Crefname{theorem}{Thm.}{Thms.}
\Crefname{thm}{Thm.}{Thms.}
\Crefname{paragraph}{Sec.}{Secs.}
\Crefname{appendix}{Appx.}{Appxs.}
\Crefname{corollary}{Cor.}{Cors.}
\Crefname{example}{Ex.}{Exs.}
\Crefname{proposition}{Prop.}{Props.}
\Crefname{remark}{Rem.}{Rems.}

\newcommand{\app}{\mathop{\cdot}}
\newcommand{\Expr}[1]{E_{#1}}
\newcommand{\encode}[2]{c_{\lambda^{#1}.#2}}
\newcommand{\code}[1]{e_{#1}}

\newcommand{\postmod}{\Diamond}

\renewcommand{\P}{\mathcal{P}\!}
\newcommand{\C}{\mathbb{C}}
\newcommand{\K}{\mathbb{K}}
\newcommand{\id}{\mathrm{id}}

\newcommand{\mca}[0]{\mathbb{A}}

\renewcommand{\encode}[2]{\left\langle\lambda^{#1}.#2\right\rangle}
\newcommand{\absterm}[2]{\left|\lambda^{#1}.#2\right|}

\newcommand{\interp}[1]{\left\llbracket #1 \right\rrbracket}

\newcommand{\letin}[3]{\mathsf{let}\ #1\  \Leftarrow #2 \ \mathsf{in} \ #3}

\newcommand{\evalstate}{\vartriangleright}
\newcommand{\applystate}{\blacktriangleleft}

\newcommand{\obj}[1]{\mathbf{Obj}\left(#1\right)}
\renewcommand{\hom}[3]{#1\left(#2,#3\right)}

\newcommand{\termobj}{\mathbf{1}}
\newcommand{\tuple}[1]{\left\langle #1 \right\rangle}
\newcommand{\pair}[2]{\tuple{#1 , #2}}

\newcommand{\underlying}[1]{\left|#1\right|}

\newcommand{\defeq}{:=}
\newcommand{\setcat}{\mathbf{Set}}
\newcommand{\codef}{\mathsf}

\newcommand{\restrict}[1]{\overline{#1}}

\newcommand{\aparr}{\circledast}

\newcommand{\monad}{M}

\newcommand{\return}[1]{\left[#1\right]}

\newcommand{\scomb}{\mathsf{s}}
\newcommand{\kcomb}{\mathsf{k}}
\newcommand{\purefun}{J}
\newcommand{\kcode}{\code{\mathsf{K}}}
\newcommand{\scode}{\code{\mathsf{S}}}
\newcommand{\transition}{\hookrightarrow}
\newcommand{\kcodeA}[1]{\code{\mathsf{K}\left(#1\right)}}
\newcommand{\scodeA}[1]{\code{\mathsf{S}\left(#1\right)}}
\newcommand{\scodeB}[2]{\code{\mathsf{S}\left(#1,#2\right)}}

\newcommand{\real}[1]{\mathcal{R}_{#1}}
\newcommand{\evidence}[1]{\overset{#1}{\rightarrow}}

\newcommand{\codeid}{\code{\codef{id}}}

\newcommand{\codeP}[1]{\code{\codef{p_{#1}}}}
\newcommand{\codepA}{\codeP{1}}
\newcommand{\codepB}{\codeP{2}}
\newcommand{\codecurry}[1]{\elambda{#1}}
\newcommand{\codeuncurry}[1]{\rho \, #1}
\newcommand{\codeeval}{\code{\codef{eval}}}

\newcommand{\codepair}[2]{\epair{#1}{#2}}

\newcommand{\evrel}[3]{#1 \evidence{#2} #3}
\newcommand{\evexpand}[3]{\forall c \in \mca . #1\left(c\right) \leq \after{r}{#2 \app c}{#3\left(r\right)}}

\newcommand{\pole}{\bot\!\!\!\bot}
\newcommand{\llangle}{\langle\!|}
\newcommand{\rrangle}{|\!\rangle}

\newcommand{\xle}[1]{\xrightarrow{\!#1\!}}
\newcommand{\eid}{e_{\mathtt{id}}}

\newcommand{\etrue}{e_{\scriptscriptstyle\top}}

\newcommand{\epair}[2]{\llangle #1, #2 \rrangle}

\newcommand{\efst}{e_{\mathtt{fst}}}

\newcommand{\esnd}{e_{\mathtt{snd}}}

\newcommand{\elambda}[1]{\lambda #1}

\newcommand{\eeval}{e_{\mathtt{eval}}}

\newcommand{\imp}{\supset}
\newcommand{\power}{\mathcal{P}}

\newcommand{\ecomp}[2]{#1 \mathop{;} #2}

 \newcommand{\after}[3]{ {\langle \!\!\mathrel{\raisebox{-4pt}{$\diamond$}}}\, #1 \leftarrow #2  {\mathrel{\raisebox{-4pt}{$\diamond$}}\!\!\rangle}\,   #3 }

\newcommand{\leql}[1]{\underset{\scriptscriptstyle{ #1 }}{\leq} }

\newcommand{\UFam}{\mathsf{UFam}}

\newcommand{\asmmap}[2]{\|#2\|_{#1}}
\newcommand{\Asm}{\ensuremath{\mathbf{Asm}}}
\newcommand{\Set}{\ensuremath{\mathbf{Set}}}
\newcommand{\op}{\mathrm{op}}
\newcommand{\ef}{\ensuremath{\mathcal{E\!F}}}
\newcommand{\Phirel}{\Phi_{\mathsf{ev}}}

\newcommand{\A}{\mathcal{A}}
\newcommand{\Trip}{\mathcal{T}}
\newcommand{\codetuple}[2]{\left\langle #1 , #2 \right\rangle}
\newcommand{\effshell}{\ensuremath{\mathcal{MC}}_M}

\newcommand{\emeet}{\wedge}

\newcommand{\meet}{\sqcap}
\newcommand{\join}{\sqcup}
\newcommand{\haimp}{\sqsupset}
\newcommand{\infimum}{\bigsqcap}
\newcommand{\supremum}{\bigsqcup}

\newcommand{\pred}{\phi}
\newcommand{\predI}[1]{\pred_{#1}}
\newcommand{\predA}{\predI{1}}
\newcommand{\predB}{\predI{2}}
\newcommand{\predC}{\predI{3}}
\newcommand{\preds}{\vec{\pred}}

\newcommand{\separator}{\mathcal{S}}

\newcommand{\singleton}{\star}

\newcommand{\param}{\mathbf{P}}

\newcommand{\pairing}[1]{\left\langle #1 \right\rangle}

\newcommand{\sub}[1]{\!\left\{ #1 \right\}}

\newlength{\fstlen}
\newlength{\sndlen}
\newlength{\thdlen}
\newlength{\mysep}
\newlength{\thesep}
\newcommand{\twothings}[2]{
  \settowidth{\fstlen}{$#1$}
  \setlength{\mysep}{0.5\textwidth}
  \addtolength{\mysep}{-\fstlen}
  \addtolength{\mysep}{-15pt}
\[
  #1
  \hspace{\mysep}
  #2
\]
}

\newcommand{\expr}{p}

\newcommand{\negspace}{\vspace{-0.5em}}
\begin{document}

\maketitle

\begin{abstract}
Partial Combinatory Algebras (PCAs) provide a foundational model of the untyped $\lambda$-calculus and serve as the basis for many notions of computability, such as realizability theory. 
However, PCAs support a very limited notion of computation by only incorporating non-termination as a computational effect. 
To provide a framework that better internalizes a wide range of computational effects, this paper puts forward the notion of Monadic Combinatory Algebras (MCAs). 
MCAs generalize the notion of PCAs by structuring the combinatory algebra over an underlying computational effect, embodied by a monad. 
We show that MCAs can support various side effects through the underlying monad, such as non-determinism, stateful computation and continuations. 
%
%
We further obtain a categorical characterization of MCAs within Freyd Categories, following a similar connection for PCAs.
Moreover, we explore the application of MCAs in realizability theory, presenting constructions of effectful realizability triposes and assemblies derived through evidenced frames, thereby generalizing traditional PCA-based realizability semantics.
The monadic generalization of the foundational notion of PCAs provides a comprehensive and powerful framework for \emph{internally} reasoning about effectful computations, 
paving the path to a more encompassing study of computation and its relationship with realizability models and programming languages.
\end{abstract}

\lcnote{cite and connect to Paul Levy}

\section{Introduction}

Partial Combinatory Algebras (PCAs) offer a foundational algebraic model for the untyped $\lambda$-calculus, which underpins various notions of computability~\cite{feferman1975language,hofstra2004partial}.
PCAs have been widely applied in areas such as realizability interpretations of logic and type theory~\cite{kleene1945interpretation,van2008realizability}, as well as in the construction of categories of assemblies and sheaves~\cite{FREY20192000}. 
A key feature of PCAs is their partiality, meaning that the application of a program to an input is defined as a partial operator, i.e. not all inputs are guaranteed to produce valid outputs. 
This property allows PCAs to naturally support non-termination as a computational effect. 
However, PCAs are inherently limited in their ability to model other critical effects, such as nondeterminism, stateful computation, or exceptions.
%
Thus, despite their foundational significance,  traditional PCAs are inherently limited in their computational scope and thus restricted in their utility to support modern computation, which often demands a broader treatment of side effects.

While extensions of PCAs, such as PCAs with errors (PCAE),  with choice (PCAC), with partiality and exceptions (PCAX) \cite{cohen2019effects,CohMiqTat21}, have been developed to accommodate specific effects, these approaches often rely on incorporating additional structures or operations into the underlying algebraic framework. 
%
However, the ability to internally support various computational effects is crucial for the development of robust and expressive computational theories.
Therefore, our goal is to develop a  unified framework that internalizes a wide range of computational effects into combinatory algebra  in a manner as natural as the support for non-termination in PCAs. 
This 
will allow for a more encompassing study of computation and new insights in the wide range of application domains of combinatory algebras.


A common categorical device for analyzing computational effects in the study of programming languages is through monads~\cite{moggi1991notions,wadler1995monads}.
Concretely, (strong) monads are a special kind of functors that allow the composition of Kleisli morphisms, i.e., morphisms where the target is an object in the image of the functor.
Using monads, a procedure that takes a value of type $A$ to a value of type $B$, while possibly invoking some computational effect, can be modeled by a morphism from $A$ to $M B $, where $M$ is some monad which embodies the effect.
For example, a common way to model nondeterministic computation is to use the powerset monad, which assigns to each set its set of subsets, so the output subset is the set of possible values the computation may yield.


To address the limitations of PCAs, this paper introduces Monadic Combinatory Algebras (MCAs), a novel generalization of PCAs designed to encapsulate a broader spectrum of computational effects by integrating monads into the combinatory algebra framework. 
%
Leveraging the monadic structure enables MCAs to internalize various side effects, such as non-deterministic computations, stateful computations, exceptions, and more.
Hence, MCAs can go beyond partiality and seamlessly integrate a wide range of effects through the underlying monad.
In doing so they provide a powerful framework for \emph{internally} supporting effectful computations, which makes them particularly valuable in modern computational contexts.
%
In fact, we show that MCAs can be instantiated to capture known effectful frameworks, such as non-determinism, stateful computation, continuations and parametric realizability.


We further provide a categorical characterization of MCAs within the context of Freyd categories, in the spirit of a similar characterization previously known of PCAs in Turing categories~\cite{COCKETT2008}. Specifically, we define a notion of a combinatory object in a Freyd category, and show that for monads in the category of sets, $\setcat$, the algebraic notion of an MCA  precisely coincides with the categorical notion of a combinatory object in the Kleisli category $\setcat_{\monad}$ of $\monad$ over $\setcat$, which forms a Freyd category.

Finally, we explore the application of MCAs in the context of realizability theory, building upon the foundational role of PCAs in traditional realizability models~\cite{van2008realizability}.
Effectful realizability and more generally extensions of the Curry-Howad paradigm using effects to account for more advanced reasoning principles have been extensively studied over the last decades, with various approaches exploring its theoretical and practical implications~\cite{BerBezCoq98,Krivine09,Boulier+Pedrot+Tabareau:cpp:2017,Pedrot+Tabareau:esop:2018,Pedrot+Tabareau+Fehrmann+Tanter:icfp:2019,Pedrot+al:lics:2020,Pedrot+Tabareau:popl:2020}. Among these, the role of monads has been pivotal in modeling and managing effects, offering a structured and algebraic perspective on effectful computation. 
Separately, the combinatorial approach to realizability has also garnered some new attention, emphasizing 
the benefits of working with combinators for algebraic purposes~\cite{Hofstra06relative,Streicher13krivine,FerEtAl17ordered,Speight24}. 
Yet, regarding effects these works were mostly focused on the particular case of Krivine realizability, i.e.,  computations manipulating continuations~\cite{Streicher13krivine,FerEtAl17ordered}. 
In this work, we present the first integration of these perspectives, combining the algebraic insights of monads with the combinatorial methodology to offer a unified framework for realizability and monads.

Two standard realizability models stemming from PCAs are given by topoi and assemblies. 
From a categorical perspective, starting from a PCA one would define a tripos, and then use the tripos-to-topos construction to get a topos~\cite{hyland1980tripos}. 
Properties of this realizability topos can be studied more easily directly into its subcategories of assemblies, which again can be defined over any PCA. 
Cohen \emph{et al.} observed that effectfull realizability models could be defined in a uniform way using a structure called \emph{evidenced frame}, factorizing the usual construction of a tripos from a PCA and making it compatible with effectful computational system~\cite{CohMiqTat21}. 
Leveraging this, we here define evidenced frames over general MCAs, and even extend this approach to consider assemblies over any evidenced frame. This illustrates how traditional approaches to realizability based on PCAs easily generalize to MCAs.                    
Moreover, the utility of the MCA framework is highlighted through several examples of known realizability models that arise naturally through the uniform MCA-based constructions. 

\paragraph*{Outline and Main contributions.} 
\begin{itemize}[leftmargin=*]
 \item \Cref{sec:background} reviews the necessary background on PCAs and monads.
\item \Cref{sec:mca} formally defines the novel structure of Monadic Combinatory Algebras
, extending the foundational structure of PCAs to support a wide range of computational effects through underlying monads. 
The utility of the MCA framework is illustrated by modeling various computational effects 
(\Cref{sec:examples}), and the categorical characterization of MCAs is established within Freyd Categories~(\Cref{sec:turing}). 
\item \Cref{sec:realizability} shows how usual realizability semantics (namely triposes and assemblies) can be generalized to MCAs, via uniform constructions factoring through evidenced frames. 
\item \Cref{sec:conc} concludes with a summary of contributions and directions for future research.
\item We supplement our development with an accompanying Rocq mechanization~\cite{Coqproofs}, hyperlinked with this paper via clickable \coqdoc{} icons.

%
\end{itemize}





\section{Background}\label{sec:background}

This section briefly overviews PCAs and monads, establishing key notations and conventions.

\subsection{Partial Combinatory Algebras}\label{subsection:PCA}

Partial Combinatory Algebras (PCAs) provide a foundational model for the untyped $\lambda$-calculus and underlie many frameworks in computability, such as realizability theory. 
They are a generalization of combinatory algebras that allows for partial functions.
That is, PCAs extend combinatory algebras by incorporating the computational effect of non-termination. 

The definition of partial combinatory algebra is based on  \emph{partial applicative structures}.
\begin{definition}[Partial Applicative Structure]
A partial applicative structure is a set $\mca$ of ``codes'' equipped with a partial binary ``application'' operator on $\mca$:
$ \left(-\right) \app \left(-\right) : \mca \times \mca \rightharpoonup \mca. $
%
\end{definition}

We use~$c_f \app c_a \downarrow c_r$ to denote $c_r$ being the (successful) result of the application~$c_f \app c_a$, and  $c_{f} \app c_{a} \downarrow$ to denote that there is a code $c_{r} \in \mca$ such that $c_{f} \app c_{a} \downarrow c_{r}$.



A PCA is then defined as a partial applicative structure that is ``functionally complete'', meaning there is a way to encode application expressions 
(such as $(c_1 \app (c_2 \app c_3)) \app (c_4 \app c_5)$) 
with $n$~free variables as individual codes accepting $n$~arguments through applications. 
That is, we ensure that every formal expression involving variables, codes, and application, can be ``internalized'' in a similar vein to the abstraction done in $\lambda$-calculus.
This ensures the necessary expressiveness for modeling computational systems like the $\lambda$-calculus.

To present the formal definition, we first formalize expressions~$e$ with numbered free variables~$i \in \mathbb{N}$. 
To make the formalism easier to mechanize, we opt for lexical addresses (de Bruijn level notation), rather than nominal variables.
Lexical addresses are numbers representing the index of a formal parameter within the body of a bound expression.
For example, given a function $f(x_{0},x_{1},x_{2})$, $2$ is the lexical address of $x_{2}$ because it is the third parameter\emnote{third but 2, can't we rephrase}, so instead of writing $x_{2}$ within the body of $f$, one can use the number $2$.
%
$$
e  {}::={}  i \in \mathbb{N} \mid c \in \mca \mid e \bullet e \qquad\qquad
E_n(\mca)  {}::={}  \{ e \mid \text{all $i$s in $e$ are $< n$}\} 
$$
Term application $\bullet$ is left associative. Elements of $E_{0}\left(\mca\right)$ are called \emph{closed} , while for every $n>0$, elements of $E_{n}\left(\mca\right)$ are called \emph{open}.
Next, we define substitution~$e\sub{c_a}$ and evaluation to expressions~$e \downarrow c_r$ as follows: 
\negspace
$$\begin{array}{c@{\hspace{1cm}}c}
    \begin{tabular}{c|c}
         $e $ & $e\sub{c_a}$ \\
         \hline
         $0$ & $c_a$ \\
         $i + 1$ & $i$ \\
         $c$ & $c$ \\
         $e_{1} \bullet e_{2}$ & $e_{1}\sub{c_a} \bullet e_{2}\sub{c_a}$
    \end{tabular}
&
 \begin{tabular}{c}
$\inferrule*{ }{c \downarrow c}$ \\
~\\
$\inferrule*{e_f \downarrow c_f \qquad e_a \downarrow c_a \qquad c_f \app c_a \downarrow c_r}{e_f \bullet e_a \downarrow c_r} $
    \end{tabular}
\end{array}
$$
\label{exp+sub}
\negspace
%

\begin{definition}[Partial Combinatory Algebra]
A \emph{partial combinatory algebras (PCA)} is a partial applicative structure $\mca$ 
such that for every expression~$e \in E_{n+1}\left(\mca\right)$ there is a code~$\encode{n}{e} \in \mca$, satisfying the following laws:
\begin{align*}
\left\langle \lambda^{n+1} . e \right\rangle \app c_{a} \downarrow \left\langle \lambda^{n} . e \sub{c_{a}} \right\rangle 
\qquad \qquad
\left\langle \lambda^{0} . e \right\rangle \app c_{a} \downarrow c_{r} \iff e \sub{c_{a}} \downarrow c_{r}
\end{align*}
\end{definition}
The code $\encode{n}{e}$ is essentially the closure of the open expression $e$,   embodying the $\lambda$-calculus term binding the $n+1$ free variables in~$e$.
Note that while we use this $\lambda$-construct to specify the result of a partial application, the above definition is equivalent (under certain assumptions) to the perhaps more traditional PCA definition~\cite{hofstra2004partial}.

\subsection{Monads}

Roughly speaking, monads are used to relate a computational effect with a specific endofunctor, say $\monad$, and consider effectful programs that return results of type $A$ in context $X$ as
morphisms $f \in \hom{\C}{X}{\monad A}$. 
A \emph{monad} $\monad$ over category $\C$ is a functor $\monad : \C \rightarrow \C$, equipped with two natural transformations $\eta : \termobj \Rightarrow \monad$ and $\mu : \monad \monad \Rightarrow \monad$, satisfying equational laws, such that it forms a monoid object in the category of endofunctors over $\C$.

In the general case, interpreting computational effects in arbitrary categories requires a \emph{strong monad} rather than just a monad~\cite{moggi1991notions}. However, here, for simplicity, we focus on a set theoretic settings, where every monad over the category of sets is a strong monad.
Hence, we sometimes use a notation similar to the one from $\lambda_{c}$ calculus~\cite{moggi1989computationallamdba}, where $\return{a}$ stands for $\eta\left(a\right)$, often called `return', and $\letin{x}{m}{n}$ stands for $\mu \circ \monad\left(\lambda x . n\right)\left(m\right)$, often called `bind'.


Given a category $\C$ and a monad $\monad$ on $\C$, the \emph{Kleisli category} $\C_{\monad}$ 
has the same objects as \( \C \), 
while its morphisms are Kleisli morphisms, i.e., $\hom{\C_{\monad}}{A}{B} := \hom{\C}{A}{M B}$. 
The identity morphism over an object \( A \) is $\eta_{A} \in \hom{\C}{A}{\monad A}$. 
Composition in \( \C_{\monad} \) is defined for $f \in \hom{\C_{\monad}}{B}{C}$ and $g \in \hom{\C_{\monad}}{A}{B}$ by $\mu \circ \monad f \circ g \in \hom{\C}{A}{\monad C} $.

\section{Monadic Combinatory Algebras}\label{sec:mca}

This section presents the generalization of the notion of PCA to one that supports a wide range of computational effects. 
%
%
For this, we employ the standard structure for describing and reasoning about general effects: (strong) monads~\cite{moggi1991notions}. 
This categorical model has the benefit of preserving many useful properties of functions while allowing a wider variety of models.
Thus,~\Cref{sec:mcadef} defines Monadic Combinatory Algebra (MCAs), generalizing PCAs by encapsulating computational effects via monads,~\Cref{sec:examples} demonstrates their versatility in embedding computational effects, and~\Cref{sec:turing} provides their categorical characterization.


\subsection{The Effectful Algebra}\label{sec:mcadef}

Just as PCAs describe the application as a \emph{partial} operator, our generalized notion describes the application as an \emph{effectful} operator, parameterized by some $\mathbf{Set}$ monad.
%


\begin{definition}[Monadic Applicative Structure \coqdoc{MAS}]
Given a $\mathbf{Set}$ monad $M$, a \emph{Monadic Applicative Structure} (MAS) over $M$ is a set of ``codes''
$\mca$ with an application Kleisli function:
$\left(-\right)\app\left(-\right) : \mca\times\mca \rightarrow M\mca$.
\end{definition}

Terms with variables and substitution follow the PCA grammar.
We think of $\left(-\right)\app\left(-\right)$ as a function that takes two codes and returns a `computation' of a code.
Intuitively, codes are either the encoding of programs, or encodings of inputs of programs.
Using the same encoding for both allows us to discuss programs that take encodings of other programs as inputs. 
That is, viewing a code $c_{f}$ as the encoding of a program, and a code $c_{a}$ as the encoding of its input, $c_{f} \app c_{a}$  denotes the computation executed by running the program encoded by $c_{f}$ on $c_{a}$ as an input.
So far, this is identical to the PCA settings.
However, whereas the PAS application function can only produce a code (if anything at all),  the MAS application function can produce any computational effect describable by a monad. 
%



To go from the applicative structure to an algebra, we must employ some evaluation map that given a term, `computes' the `value' of the term.
%
%
The key to the evaluation is the way in which it connects to the underlying monad~\cite{wadler1990comprehending}.
%
%
Intuitively, the evaluation of a code simply returns the code, without exhibiting any effectful behavior.
Application terms, on the other hand, evaluate using the monadic bind, depending on the specific strategy, and may exhibit the computational effect carried by the monad.


%
In principle, MCAs are orthogonal to the specifics of the evaluation strategy. 
That is,  we can define a notion of an MCA based on various evaluation strategies. 
In this paper, to conform with traditional PCA, 
we opt to commit to the common Call-by-Value (CbV) evaluation strategy, codified in the following evaluation model.
%

\begin{definition}[CbV Evaluation \coqdoc{eval}]\label{def:cbv-eval}
Evaluation $\nu$ is 
defined by induction on $E_{0}\left(\mca\right)$ as follows:
$$
\nu\left(c\right) := \return{c} \qquad \qquad
       \nu\left(e_{f} \bullet e_{a}\right) :=
       \letin{c_{f}}{\nu\left(e_{f}\right)}{\letin{c_{a}}{\nu\left(e_{a}\right)}{c_{f} \app c_{a}}}
$$
\end{definition}
Note that evaluation is reflected by an \emph{equality} in $\mca$. 
However, this could be generalized to an ordered setting following~\cite{Hofstra06relative,FerEtAl17ordered}.

Next, we turn to the definition of monadic combinatory algebras.
To ensure they can support general computation, they must be at least as computationally powerful as the $\lambda$-calculus.
This is done through a similar mechanism to how PCAs are defined.
But, unlike PCAs, which permit only purely functional behavior (up to non-termination), monadic combinatory algebras accommodate any computational effect representable via monads.
 
\begin{definition}[Monadic Combinatory Algebra \coqdoc{MCA}]\label{mca}
A Monadic Combinatory Algebra (MCA) is a monadic applicative structure $\mca$ such that for every expression~$e \in \Expr{n+1}\left(\mca\right)$ there is a code~$\encode{n}{e} \in \mca$ satisfying the following laws:
$$\begin{array}{@{}llcl}
\forall n \in \mathbb{N}. \forall e \in \Expr{n+2}\left(\mca\right). \forall c \in \mca. & \encode{n+1}{e} \app c &=& \eta \left( \encode{n}{e\sub{c}} \right) \\
	\forall e \in \Expr{1}\left(\mca\right).\forall c \in \mca .& \encode{0}{e} \app c &=& \nu \left( e\sub{c} \right)
\end{array}$$
\end{definition}



The \emph{abstraction} assignment of codes to expressions turns open terms into codes that internalize their evaluation after substitution using the monad's return and bind.
Principally, the MCA laws state that it ``delays'' the evaluation of an open term it encloses until all free variables in it are substituted by codes.
If there is more than one free variable in the term, using $\eta$ means the application of the abstracted term to a code only substitutes the outermost parameter with the argument.
However,  if the term has exactly one free variable,  applying it to the argument obtains a closed term, which is then immediately evaluated.

Terms that only involve codes given by the abstraction operator alone are called \emph{pure terms}.
The MCA definition closely resembles that of a PCA; in fact, when evaluating pure terms, the behavior remains identical to a PCA, exhibiting no computational effects beyond non-termination.
Thus, computational effects arise only through additional codes not definable via abstraction.
For example, when describing non-deterministic computation, using the powerset monad, the evaluation of pure terms will always result in a set with at most one value, while using the monad allows $\mca$ to have codes which, when applied, yield a set with multiple values (see~\Cref{sec:examples}).
However, the monad still plays an important role when evaluating pure terms, as it allows interpreting a non-terminating evaluation as special elements within the set of computations, denoting the absence of a value. 
%
Concretely, a PCA is obtained from an MCA by instantiating the monad with the sub-singleton monad, i.e.  $M A$ is the set of subsets of $A$ in which all elements are equal.
\[
  \monad A = \{S\subseteq A \mid \forall x_{1},x_{2}\in S. x_{1} = x_{2} \}
 \qquad\qquad
\eta_A\left(x\right) = \{x\}
 \qquad\qquad
 \mu_A\left(m\right) = \bigcup_{X\in m} X
\]

\begin{proposition}\label{pca_mca}
    PCA is a special case of an MCA.
\end{proposition}

\Cref{mca} is equivalent to the (perhaps more familiar) formalization of combinatory completeness via the $\scomb$ and $\kcomb$ combinators~\cite{van2008realizability}.
That is, we could have alternatively, defined an  MCA as a MAS with $\scomb$ and $\kcomb$ combinators, because $\scomb$ and $\kcomb$ are simply encodings of particular expressions that are sufficient to ensure that all expressions can be encoded, and, in the converse direction, the above MCA laws essentially ensure their existence.
Note, however, that the characterizing axioms for these combinators in the discourse are tailored to a non-effectful behavior of the calculus. 
Standardly, the axioms of $\scomb$ and $\kcomb$  require all their partial applications to be defined.
Usually written as $\kcomb \app c_{1} \downarrow$, $\scomb \app c_{1} \downarrow$, and $\scomb \app c_{1} \app c_{2} \downarrow$ for any $c_{1}$ and $c_{2}$.
In the context of MCAs, however, it is not sufficient for the partial applications to be defined, but they also have to be pure, not triggering any effects beyond the evaluation of their value.
This constraint ensures $\scomb$ and $\kcomb$ do exactly the same thing they do in PCAs, and nothing else, regardless of the underlying monad.
For monads, purity is represented with the monad unit.
Hence, the monadic version of the axioms requires the partial applications of $\scomb$ and $\kcomb$ to return the application of the monad unit over particular codes.

\begin{proposition}[\coqdoc{skmca}]\label{prop:SK}
    An MCA is equivalent to a MAS with 
  codes $\scode , \kcode , \scodeA{c_{1}} , \scodeB{c_{1}}{ c_{2}}$, and $\kcodeA{c_{1}}$ for any two codes $c_{1}$, $c_{2}$, satisfying the following axioms:
$$
\begin{array}[t]{l@{~}c@{~}l}
    \scode \app c_{1} &= &\return{\scodeA{c_{1}}}\\
    \kcode \app c_{1} &=& \return{ \kcodeA{c_{1}}}\\
\end{array}
\qquad~~
\begin{array}[t]{l@{~}c@{~}l}
    \scode \app c_{1} &= &\return{\scodeA{c_{1}}}\\
    \kcodeA{c_{1}} \cdot c_{2} &= &\return{c_{1}}\\
\end{array}
\qquad~~
\begin{array}[t]{l@{~}c@{~}l}
    \scodeB{c_{1}}{c_{2}} \app c_{3} &= &\nu\left(\left(c_{1} \bullet c_{3}\right) \bullet \left(c_{2} \bullet c_{3}\right)\right)
\end{array}
$$
\end{proposition}

\subsection{MCA Instances}\label{sec:examples}

\begin{figure}[t]
    \centering
\scalebox{1}{
$\begin{array}{@{}l@{\hspace{0.3cm}}l@{\hspace{0.35cm}}l@{\hspace{0.35cm}}l}
\hline
\textbf{Comb. Alg.} & \textbf{Monad} ~\monad A & \textbf{return}~\eta_A\left(x\right)  & \textbf{bind}~\mu_A\left(m\right) \\
\hline
\textit{Partial} & \{X\subseteq A \mid 
\forall x,y\in X. x = y \}
&
\{x\}
&
\bigcup_{X\in m} X\\
\textit{Relational} & \P\left(A\right) & {\{x\}}
&
\bigcup_{X\in m} X \\
\textit{Stateful} & {
  \Sigma\rightarrow \P\left(\Sigma\times A\right)
}
&{
  \lambda \sigma . \{(\sigma,x)\}
}
&{
\lambda \sigma . \bigcup_{\left(\sigma',f\right)\in m\left(\sigma\right)}\!\!f\left(\sigma\right)
}\\
\textit{CPS} & {
 \left(A \rightarrow R\right) \rightarrow R
}&{
  \lambda k.k\left(x\right)
}&{
\lambda k . m\left(\lambda g.g\left(k\right)\right)
}\\
\textit{Parameterized} & {
 \param \rightarrow \{X\subseteq A \mid 
 \forall x,y\in X. x = y \}
 }&{
 \lambda p . \left\{ x \right\}
 }&{
 \lambda p . 
 \bigcup \left\{ g\left(p\right) \mid g \in m\left(p\right) \right\}
}
\end{array}$
}
    \caption{MCA Instances}
    \label{fig:instances}
\end{figure}

As demonstrated, PCAs are a special case of MCAs. 
However, the generalized structure also encompasses many common effectful structures proposed in the literature. This section illustrates how those can be derived from an MCA by instantiating the underlying monad.


\subsubsection{Relational Combinatory Algebra}\label{RCA}
\emph{Relational Combinatory Algebras} (RCAs) were defined in~\cite{cohen2019effects,CohMiqTat21}
to account for (demonic) non-determinism.
Concretely, they were used to show that while all realizability models stemming from PCAs model the principle of Countable Choice, realizability models based on RCAs can, in fact, model the negation of the principle.

RCAs correspond to an MCA with $\monad$ being the \emph{powerset} monad, as in~\Cref{fig:instances}. The powerset monad captures non-deterministic computations by considering the subset of possible results. 
When $\mca$ is an RCA, it may have codes such as $\code{\mathsf{flip}}$, which nondeterministically returns either $\encode{1}{0}$ or $\encode{1}{1}$ whenever it is applied.
That behavior is defined by representing the set of possible return values when describing the application of $\code{\mathsf{flip}}$
:
$ \code{\mathsf{flip}} \app c = \left\{ \encode{1}{0} , \encode{1}{1} \right\} $.

\subsubsection{Stateful Combinatory Algebra}\label{SCA}
\emph{Stateful Combinatory Algebras} (SCAs) were defined in~\cite{cohen2019effects,CohMiqTat21} as a stateful extension of RCAs.
SCAs were used to memoize non-deterministic
computations and recover a realizer of Countable Choice even in the presence of nondeterminism.

SCAs correspond to MCAs where $\monad$ is 
the \emph{powerset state} monad, as given in~\Cref{fig:instances}. 
This allows taking a code in a given state and returning a set of all possible pairs of results in new states.
In~\cite{cohen2019effects,CohMiqTat21}, a variant of the \emph{increasing} state monad was used, which is a submonad of the powerset state monad:
$
\monad A = \left\{ m:\Sigma\rightarrow \P\left(\Sigma\times A\right)\
 \mid\ \forall\sigma_{0}\in\Sigma.\forall\left(\sigma_{1},x\right)\in m\left(\sigma_{0}\right).\sigma_{0}\leq\sigma_{1}\right\}
$.


When $\mca$ is an SCA, it may have codes such as $\code{\mathsf{get}}$ and $\code{\mathsf{inc}}$, implementing a counter.
For simplicity, we take states to be natural numbers.
When $\code{\mathsf{get}}$ is applied to a code $c$, it ignores $c$ and returns the Church numeral representing the current state of the counter, leaving the state intact.
When $\code{\mathsf{inc}}$ is applied to a code $c$, it increments the counter and returns $c$.
\twothings{
    \code{\mathsf{get}} \app c = \lambda n . \left\{ \left(\overline{n} , n\right) \right\}
    }{
    \code{\mathsf{inc}} \app c = \lambda n . \left\{ \left(c , n+1\right) \right\}
    }

\subsubsection{CPS Combinatory Algebra}\label{example:CPS}
The double-negation translation~\cite{godel1933intuitionistic},  relating classical logic with intuitionistic logic, points to a connection between classical logic and the continuation monad. 
This connection has been extensively studied, particularly in the context of classical realizability~\cite{Krivine09}.
However, most work on classical realizability is based on Krivine abstract machines rather than combinatory algebras, which is an entirely different model of computation.
By utilizing MCAs with the continuation monad, we can align classical realizability with a computational model akin to intuitionistic realizability.

Here, we focus on Continuation-Passing-Style (CPS), which is a style of programming in which control is explicitly passed through continuation functions. Thus, instead of returning results directly, functions in CPS receive an extra argument: a continuation function that specifies what to do next with the result.
A \emph{CPS Combinatory Algebra} (CPSCA) is an MCA where the underlying monad is the CPS monad, as described in~\Cref{fig:instances}.\footnote{In general, classical realizability is not constructed via combinatory algebras, but there are similar structures, e.g.~\cite{FerEtAl17ordered} which uses ordered combinatory algebras, or~\cite{krivine2011realizability} which uses abstract machines.}
The CPS monad allows for composable and reusable continuations, i.e., it models computations with direct access to the call stack, enabling non-trivial control flow manipulation.
In the above definition, $R$ can be any set, representing the ultimate results of the whole computation.

When $\mca$ is a CPSCA, it may have codes such as $\code{\mathsf{cc}}$ and $\code{\mathsf{K}_{u}}$, which save and replace the current continuation:
\negspace
\twothings{
    \left(\code{\mathsf{cc}} \app c_{a}\right)\left(u\right) = \left(c_{a} \app \code{\mathsf{K}_{u}} \right)\left(u\right)
}{
    \left(\code{\mathsf{K}_{u}} \app c_{a}\right)\left(u'\right) = u\left(c_{a}\right)
}


\agnote{I feel like the reader/reviewer won't understand the point of the defunctionalization paragraph, so we should probably either properly motivate it, or remove it}\emnote{I'll try to see if it can be improved, but I think having a ref to a nice paper giving a methodology to go from one to the other is rather a good point for us (and it also justifies just giving the machine and claiming it comes from the evaluation as functions without having to give more details).}
Using CPSCAs, the definitions of evaluation and application can be seen as a CPS form of an evaluator for PCAs.
By defunctionalization, this leads directly to an \emph{eval/apply} abstract stack machine, providing operational semantics for PCAs~\cite{danvy2004evaluation}.
The machine has a stack, which can hold either codes $c$ tagged with $\boldsymbol{v}\left(c\right)$, or terms $e$ tagged with $\boldsymbol{t}\left(e\right)$.
An empty stack is marked as $\emptyset$, while a nonempty stack is marked with $x : \pi$, where $x$ is the top of the stack and $\pi$ is the rest.
The machine has three states:
\begin{enumerate}[leftmargin=0.5cm]
    \item \emph{Eval} state, marked as $e \evalstate \pi$, in which the machine takes a closed term $e$, and a stack $\pi$.
    \item \emph{Apply} state, marked as $c \applystate \pi$, in which the machine takes a value $c$, and a stack $\pi$.
    \item \emph{Final} state, marked as simply a code $c$, of the final value.
\end{enumerate}
The semantics of the machine is defined via a one-step transition relation $\transition$ between states: 
\[
    \begin{array}{r@{~}c@{~}l@{\hspace{0.02cm}}c@{\hspace{0.15cm}}r@{~}c@{~}l}
        e_{f}\app e_{a} & \evalstate & \pi & \quad\transition\quad & e_{f} & \evalstate & \boldsymbol{t}\left(e_{a}\right):\pi\\
        c & \evalstate & \pi & \quad\transition\quad & c & \applystate & \pi\\
        c & \applystate & \emptyset & \quad\transition\quad & c\\    
    \end{array}~~\vrule~~
       \begin{array}{r@{~}c@{~}l@{\hspace{0.02cm}}c@{\hspace{0.15cm}}r@{~}c@{~}l}
       c_{f} & \applystate & \boldsymbol{t}\left(e_{a}\right):\pi & \quad\transition\quad & e_{a} & \evalstate & \boldsymbol{v}\left(c_{f}\right):\pi\\
        c_{a} & \applystate & \boldsymbol{v}\left(\left\langle \lambda^{0}.e\right\rangle \right):\pi & \quad\transition\quad & e\sub{c_{a}} & \evalstate & \pi\\
        c_{a} & \applystate & \boldsymbol{v}\left(\left\langle \lambda^{n+1}.e\right\rangle \right):\pi & \quad\transition\quad & \left\langle \lambda^{n}.e\sub{c_{a}}\right\rangle  & \applystate & \pi\\
    \end{array}
\]
\subsubsection{Parameterized Combinatory Algebra}\label{example:bauer}
The notion of \emph{Parameterized Combinatory Algebras} (ParCAs), introduced by Bauer and Hanson  in~\cite{bauer2024countablereals}, 
is based on a notion of computations that has access to external oracles.
Using parameterized combinatory algebras the authors constructed a model in which the set of Dedekind real numbers is countable.

A ParCA is precisely an MCA based on the \emph{subsingleton reader} monad. That is, for a set $\param$ of parameters, the parameterized monad is defined in~\Cref{fig:instances}.
A function into the subsingleton reader monad represents a partial computation, which, in addition to its input, has access to some external parameters in a set $\param$.
The uniformity of the MCA framework here is highlighted by the fact that the comprehensive algebraic axiomatization of ParCA in~\cite{bauer2024countablereals} is naturally derivable from the MCA representation.

Computations with external oracles can be modelled by simply allowing the MCA to have access to external parameters.
When $\mca$ is a ParCA, where the parameter is taken from predicates over $\mca$, that is $\param = \mca \rightarrow \left\{ 0 , 1 \right\}$, it may have a code such as $\code{\mathsf{search}}$.
When $\code{\mathsf{search}}$ is applied to another code $c$, it takes a predicate $p$ as parameter, and returns either $\encode{1}{0}$ or $\encode{1}{1}$ according to whether $c$ satisfies $p$:
\[ \left(\code{\mathsf{search}} \app c\right) = \lambda p . \texttt{if}~p\left(c\right) = 0~\texttt{then}~\left\{ \encode{1}{0} \right\}~\texttt{else}~\left\{ \encode{1}{1} \right\}
\]

\subsection{Categorical Characterization of MCAs}\label{sec:turing}
PCAs have been given a categorical representation in \cite{COCKETT2008}, by establishing their connection to
Turing categories~\cite{LONGO1990193,COCKETT2008}. 
Concretely, it was shown that PCAs are PCA-objects in the category of sets, which are essentially the categorical counterpart of the notion of combinatory completeness for PCAs. 
To obtain a similar categorical characterization for the generalized notion of MCAs which is based on arbitrary monads, we here work 
within Freyd categories~\cite{Power1997Enviroments,LEVY2003Freyd} which are an extension of the categorical framework designed specifically to model computational effects. 
Importantly, they abstract the structure of the Kleisli category of a monad by employing one category for values and another one for computations. 
Due to space limitations, we here provide the key results, leaving full details to the appendix.


\begin{definition}[Applicative Object]
    For a Freyd category $\langle \C, \K, \purefun \rangle$, an \emph{applicative object} is an object $\mca \in \obj{\C}$ equipped with an application morphism $ \aparr  \in \hom{\K}{\mca \times \mca}{\mca}$.
\end{definition}

\begin{definition}
[Computable Morphism]
\label{def:computable}
Given an applicative object $\mca$ in a Freyd category $\langle \C, \K, \purefun \rangle$ and an $0<n \in \mathbb{N}$, we say that a $\K$ morphism $f \in \hom{\K}{\mca^{n}}{\mca}$ is $\mca$-computable when, for all $k \in \left\{ 0 , \ldots , n-1\right\}$, and all $c_{1} , \ldots , c_{k} \in \hom{\C}{\termobj}{\mca}$, there is a code $\code{f\left(c_{1} , \ldots , c_{k}\right)} \in \hom{\C}{\termobj}{\mca}$ such that the following diagrams commute in $\K$ (where $\code{f}$ is $\code{f\left(c_{1} , \ldots , c_{k}\right)}$ for $k=0$).
\[\begin{tikzcd}
{\termobj \times \mca^{n-k}} & & {\mca \times \mca^{n-k}} & {\termobj} && {\mca \times \mca^{k}}\\
{\mca^{k} \times \mca^{n-k}} & {\mca^{n}} & {\mca} &&& {\mca}
\arrow["{\purefun\pairing{c_{1} , \ldots , c_{k}}_{k} \ltimes \mca^{n-k}}", from=1-1, to=2-1, swap]
\arrow["{\alpha^{k}}", from=2-1, to=2-2, swap]
\arrow["{\purefun \code{f\left(c_{1} , \ldots , c_{k}\right)} \ltimes \mca^{n-k}}"', from=1-1, to=1-3, swap]
\arrow["f", from=2-2, to=2-3, swap]
\arrow["{\aparr^{n-k}}"', from=1-3, to=2-3, swap]
\arrow["{ \purefun \pairing{\code{f} , c_{1} , \ldots , c_{k}} }_{k+1}"', from=1-4, to=1-6, swap]
\arrow["\purefun \code{f\left(c_{1} , \ldots , c_{k}\right)}", from=1-4, to=2-6, swap]
\arrow["{\aparr^{k}}"', from=1-6, to=2-6, swap]
\end{tikzcd}\]


\end{definition}


Our definition of a combinatory object relies on the notion of $\mca$-monomials, generalizing~\cite{LONGO1990193}.
An $\mca$-monomial is a morphism in $\hom{\K}{\mca^{n}}{\mca}$, for some $n \in \mathbb{N}$, defined using only projections, morphisms in $\hom{\C}{\termobj}{\mca}$, and application $\aparr$, used in an applicative order. For $n>0$ the $\mca$-monomial is called positive.

\begin{definition}[Combinatory object]\label{def:mca-obj}
An applicative object 
$\mca$ is called a \emph{combinatory object} when all positive $\mca$-monomials are $\mca$-computable.
\end{definition}

\begin{restatable}{theorem}{combObj}
Let $\monad$ be a $\setcat$ monad.
$\mca$ is an MCA over $\monad$ if and only if it is a combinatory object in $\setcat_{\monad}$.
\end{restatable}

Furthermore, the definition of a PCA-object in a cartesian restriction category~\cite{COCKETT2008}, is the exact counterpart of a combinatory object in a Freyd category, obtained by replacing morphisms of a Freyd category with their counterparts in a cartesian restriction category.

\section{MCA-induced Realizability Models}
\label{sec:realizability}

Since PCAs underpin traditional realizability models, this section explores the application of MCAs in the broader context of realizability theory. 
%
%
By incorporating computational effects, MCAs broaden realizability models, enabling a more extensive semantic framework for diverse logic systems based on computation.

\emph{Evidenced frames} propose a unifying approach to effectful realizability~\cite{CohMiqTat21}. An evidenced frame is a structure that abstracts the core components of realizability models by focusing solely on the relationship between propositions and their evidence, while omitting the computational specifics of individual models. 
The evidenced frame abstraction is complete in that any realizability tripos  (i.e. a model of higher-order dependent predicate logic) can be viewed as an evidenced frame, and there is a uniform construction generating a corresponding realizability tripos from an evidenced frame. Besides, the usual construction of a tripos from a PCA smoothly factorizes through the definition of an evidenced frame.

To provide semantic realizability models from an MCA, we build on this construction and demonstrate how an evidenced frame can be derived from it. 
This establishes a clear pathway from MCAs to realizability triposes, and in turn, through the tripos-to-topos construction~\cite{pitts2002tripos}, to a realizability topos, which is a model of (extensional, impredicative) dependent type theory (and set theory).
It further factors an alternative construction of assemblies which, in the case of PCAs, have been broadly studied in relation to the realizability topos.

The following development uses preordered sets, and in particular, complete Heyting prealgebras, for the interpretation of logic.
A preordered set $\left(\Omega,\leq\right)$ is a set $\Omega$ equipped with a reflexive and transitive ``inequality'' relation $\leq$ over $\Omega$.
A Heyting prealgebra extends a preordered set by additional algebraic operations, namely, meet, join and implication.
In the context of semantics of logic, a preordered set $\Omega$ is used as a set of ``truth values'', and formulas in the logic are interpreted as elements of $\Omega$.
That is, for $\pred$ a formula, its interpretation  $\interp{\pred}$ is an element of $\Omega$. 
The preorder relation of a Heyting prealgebra corresponds to the logical entailment $\vdash$ between formulas, while the algebraic counterpart of universal (resp. existential) quantifications is provided by meets $\meet$ (resp. joins $\join$). The bottom and top elements of $\Omega$ are denoted $\boldsymbol{0}$ and $\boldsymbol{1}$ (resp.). This will come in handy in order to give an algebraic structure to the predicates of our realizability models.

\subsection{Realizability Triposes via Evidenced Frames}
\label{sec:EF}

\begin{figure}[t]
\centering
\begin{small}
\resizebox{\textwidth}{!}{ 
\begin{tabular}{|l|l|l|l|}
\hline
\textbf{} & \textbf{Logical Const.} & \textbf{Program Const.} & \textbf{Evidence Relation} \\
\hline
\emph{Reflexivity} &  & $\eid \in E$ & $\phi \xle{\eid} \phi$ \\
\hline
\textit{Transitivity} &  & $\ecomp{}{} \in E \times E \to E$ & $\phi_1 \xle{e_1} \phi_2 \land \phi_2 \xle{e_2} \phi_3 \Rightarrow \phi_1 \xle{\ecomp{e_1}{e_2}} \phi_3$ \\
\hline
\textit{Top} & $\!\top \!\in\! \Phi$ & $\etrue \in E$ & $ \phi \xle{\etrue} \top$ \\
\hline
\textit{Conjunction} & $\emeet \!\in\! \Phi \!\times\! \Phi \!\to\! \Phi$ & 
\makecell[l]{$\epair{\cdot}{\!\cdot} \!\in\! E \!\times\! E \!\to\! E$ \\ $\efst, \esnd \in E$} &
\makecell[l]{
$\phi \xle{e_1} \phi_1 \land \phi \xle{e_2} \phi_2 \Rightarrow \phi \xle{\epair{e_1}{e_2}} \phi_1 \emeet \phi_2$
\\ $\phi_1 \emeet \phi_2 \xle{\efst} \phi_1 \quad , \quad \phi_1 \emeet \phi_2 \xle{\esnd} \phi_2$}
\\
\hline
\makecell[l]{\textit{Universal}\\\textit{Implication}} & $\imp \in \Phi \times \power(\Phi) \to \Phi$ & 
\makecell[l]{\hbox{$\elambda{} \!\in\! E \!\to\! E$} \\ $\eeval \in E$} &
\makecell[l]{
 $ (\forall \phi \in \vec{\phi}.\; \phi_1 \emeet \phi_2 \xle{e} \phi) \Rightarrow \phi_1 \xle{\elambda{e}} \phi_2 \imp \vec{\phi}$
\\ $\forall \phi \in \vec{\phi}.\; (\phi_1 \imp \vec{\phi}) \emeet \phi_1 \xle{\eeval} \phi$}
\\
\hline
\end{tabular}
}
\caption{Evidenced Frame constructs, where $\vec{\phi}\in\power(\Phi)$ and the evidence relations are universally quantified.}
\label{tab:EF}
\end{small}
\end{figure}
%
First, we recall the definition of an evidenced frame.
\begin{definition}[Evidenced Frame~\cite{CohMiqTat21} \coqdoc{EF}]\label{evidenced-frame}
An \emph{evidenced frame} is a triple $\ef = \left( \Phi, E, \mbox{$\cdot \xle{\cdot} \cdot$} \right)$, where $\Phi$ is a set of propositions,~$E$ is a set of evidence, and~\mbox{$\phi_1 \xle{e} \phi_2$} is a ternary evidence relation on $\Phi \times E \times \Phi$, 
along with the structure captured in~\Cref{tab:EF}. 
\end{definition}
 The evidenced frame setup mirrors computational processes, where evidence can be thought of as programs or computational artifacts that demonstrate logical relationships. 
In fact, evidenced frames supplement complete Heyting prealgebras, which are standard models of intuitionistic logic, by adding computational evidence for the validity of the preorder relation. 
%
Hence, while in a prealgebra,  $\predA \leq \predB$ denotes that the pair $\left(\predA,\predB\right)$ satisfies the preorder relation, in an evidenced frame the corresponding notion is $\evrel{\predA}{e}{\predB}$, meaning the triple $\left(\predA,e,\predB\right)$ satisfies the evidenced relation.
All the axioms of a complete Heyting prealgebra, described in terms of the preorder relation $\leq$, appear in an evidenced frame in an enhanced form, where $\leq$ is replaced with the evidence relation $\evrel{\cdot}{\cdot}{\cdot}$, and each axiom requires the existence of some ``evidence'' (appearing as $e$ in $\evrel{\predA}{e}{\predB}$) which uniformly witnesses the validity of the axiom.
 Evidenced frames are flexible in that they do not assume a specific equational theory which allows them to capture a broader spectrum of computational effects and unify them by focusing on uniform evidence.
    
\subsubsection{$M$-Modalities}

As shown in~\cite{CohMiqTat21}, in the standard case of a PCA, $\mca$, the corresponding evidenced frame is that in which $E$ is the set of codes $\mca$, $\Phi$ is its set of subsets $\P\left(\mca\right)$ and for every $\predA , \predB \in \P\left(\mca\right)$ and $c_{f} \in \mca$, $\evrel{\predA}{c_{f}}{\predB}$ stands for:
$ \forall c_{a} \in \mca . c_{a} \in \predA \Rightarrow \exists c_{r} .\, c_{f} \app c_{a} \downarrow c_{r}\wedge c_{r} \in \predB$. 
Our goal here is to generalize the embedding of combinatory algebras into evidenced frames and provide a uniform construction of an evidenced frame from any MCA.
However, the MCA abstraction poses a challenge. 
For PCAs, the result of the computation $c_{f} \app c_{a}$ is given by the reduction 
predicate $c_{f} \app c_{a} \downarrow c_{r}$, and then related to a subset $\predB$ through the membership predicate $c_{r} \in \predB$.
However, in an MCA, the result $c_{f} \app c_{a}$ is in $M\left(\mca\right)$, so it appears within the context of a more abstract notion of computation, and some device is needed to be able to pick the result from the computational context and relate it to a subset, or more generally, to a predicate. 
In fact, traditional realizability manipulates codes, not computations, and so our device needs to be able to extend predicates defined on values to predicates on computations, which then can be used again in the realizability setting. 
As we discuss below, we will consider predicates on a set $X$ as defined by functions from $X$ to some Heyting prealgebras $\Omega$, hence our device will essentially extend such a function to a function in $M(X) \to \Omega$.

The device appropriate for this task is called an $M$-modality, which intuitively describes a post-condition over the result of a (possibly effectful) computation.
A variant of the notion of an $M$-modality first appeared in~\cite{moggi1991notions} and then further elaborated in~\cite{pitts1991evaluation}.\footnote{\Cref{def:modality}  can equivalently be formulated as an oplax algebra $\omega : \monad \Omega \rightarrow \Omega$, similar to the T-modal operator in \cite{moggi1991notions}.
However, this will require the definition of each particular modality to operate on truth values in $\Omega$, rather than to relate value predicates with computations, which seems easier in practice. }


\begin{restatable}[$M$-modality \coqdoc{MMod}]{definition}{modality}\label{def:modality}
    Let  $M$ be a $\mathbf{Set}$ monad, $\mca$ an MCA  over $M$, and $\left(\Omega, \leq\right)$ a complete Heyting prealgebra.
    An $M$-modality over $\Omega$ is a natural transformation:
    \[\postmod_{X} : M\left(X\right) \rightarrow \left(X \rightarrow \Omega\right) \rightarrow \Omega 
    \qquad\text{and we note }\after{x}{m}{\pred\left(x\right)\defeq \postmod\left(m\right)\left(\pred\right)}
    \]
satisfying for all $A,B$, $\pred_i : A \rightarrow \Omega$, $f : A \rightarrow M\left(B\right)$, $a\in A$ and $m \in M\left(A\right)$ the following: 
\begin{description}[leftmargin=*]\setlength\itemsep{0.5em}
  \item[After-Return.] 
    $\qquad\qquad~~\pred\left(a\right) \leq \after{x}{\return{a}}{\pred\left(x\right)}  $ 

 \item[After-Bind] 
     $\qquad\qquad~~~~~~\after{x}{m}{\after{y}{f \left(x\right)}{\pred\left(y\right)}} \leq \after{y}{\letin{x}{m}{f\left(x\right)}}{\pred\left(y\right)} $

\item[Internal Monotonicity.] 
\label{eq:after-inter-mono} 
     $~~~\infimum_{c}\left(\predA\left(c\right) \haimp \predB\left(c\right)\right) \leq \after{x}{m}{\predA\left(x\right)} \haimp \after{x}{m}{\predB\left(x\right)}$ 
\end{description}
where we apply the standard precedence of quantifiers, for example $\after{x}{m}{\predA\left(x\right)} \haimp \predB\left(x\right) $ is to be read as $\left(\after{x}{m}{\predA\left(x\right)}\right ) \haimp \predB\left(x\right) $.

%
%


\end{restatable}

Intuitively, $\Omega$ is a set of truth values, so a predicate $\pred$ over a set $X$ is a function $\pred : X \rightarrow \Omega$, and to denote that $\pred$ applies to some $x \in X$, we write $\pred\left(x\right)$ as is standard.
As for the modality, we read $\after{x}{m}{\pred\left(x\right)}$ as saying that after the computation $m$ yields a value $x$ (in case it does), then $\pred\left(x\right)$ holds.
To obtain a sound logical framework, the properties of an $M$-modality ensure it is well-behaved with respect to the computational operators of the monad and the logical operators of the complete Heyting prealgebra. 

The use of Heyting prealgebras, rather than subsets, allows us to generalize the standard notion of a subset of codes to more complex subset-like structures, in particular, ones that account for the computational effects and enforce invariants over the computational behavior.
For example, in the case of stateful nondeterministic computation, as in~\Cref{SCA}, it is useful to consider subsets of pairs of codes and states, $\P\left(\Sigma \times \mca\right)$,  
and to obtain a well-behaved logic we must restrict attention to ``future-stable'' predicates, i.e.~predicates which, for every nondeterministic stateful computation, if they hold before the change of state, they keep holding for every possible mutation of the state as well.
In terms of subsets, it means the set of states has to be preordered, and instead of taking $\P\left(\Sigma \times \mca\right)$, we take $\mathcal{U}\left(\Sigma\right)^{\mca}$, where $\mathcal{U}\left(\Sigma\right)$ is the set of all upper subsets of $\Sigma$. This set has the structure of a complete Heyting algebra (and thus, a complete Heyting prealgebra) as a topological space, given by the Alexandrov topology.
However, the Heyting prealgebra structure of $\P\left(\Sigma \times \mca\right)$ alone is not enough for constructing an evidenced frame over SCAs.
For this, we need to additionally require the modality preserves implication, which is not the case in the standard angelic and demonic interpretations of nondeterminism, but is guaranteed by internal monotonicity.
\lcnote{TODO: compare this def against Pitts/... and against known names for modality axioms}

The internal monotonicity of the $M$-modality is stronger than the perhaps more well-known ``order-preserving'' property of the $T$-modality in \cite{pitts1991evaluation}.\lcnote{is this needed?}
Syntactically, internal monotonicity ensures that the modality preserves implication.
Semantically, it ensures that any property that holds before a computation keeps holding afterwards.



\newcommand{\cloop}{c_{\scalebox{0.75}{$\circlearrowright$}}}
To make sure the induced semantics are meaningful, we must verify that evidence for entailment does not exist for every pair of predicates. 
For example, in the case of PCAs, one can define the modality
$\after{x}{m}{\pred\left(x\right)} \defeq \infimum_{x \in m} \pred\left(x\right)$
corresponding to partial correctness.
However, with this modality, any code that yields no value when applied to any argument, such as $\cloop\defeq \encode{0}{\encode{0}{0 \bullet 0} \bullet \encode{0}{0 \bullet 0}}$, can be used as evidence for the entailment of  any pair of predicates. 
In particular, consider the predicates $\top = \lambda x . \boldsymbol{1}$ and $\bot = \lambda x . \boldsymbol{0}$, then $\top \evidence{\cloop} \bot$ would mean $1 \leq \infimum_{c \in \nu\left(\cloop\right)} . 0 $.
Since $\nu\left({\cloop}\right) = \emptyset$,  this statement is vacuously true, and thus we can consistently model an inconsistent theory.

To eliminate this option, we have to make sure the modality is selective enough to prevent absurd entailments 
from being evidenced.
%
To that end, we employ a similar technique to the one mentioned in~\cite{CohMiqTat21}, 
with a generalized notion of a separator for our setting.
As we shall see in~\Cref{thm:sep_consistent}, separators will indeed ensure the consistency of the induced evidenced frame.
\negspace
\begin{definition}[Separator \coqdoc{separator}]\label{def:separator}
Given 
an MCA $\mca$, a Heyting prealgebra $\Omega$ and an $M$-modality $\postmod$ over them, a separator for $\postmod$ (or a $\postmod$-separator) is a combinatory complete subset $\separator$ of $\mca$, such that, for every $c_{f} , c_{a} \in \separator$, the following ``progress'' property holds:
$ \after{r}{c_{f} \app c_{a}}{\boldsymbol{0}} \leq \boldsymbol{0} $.
\end{definition}

One common separator is the one that consists of all codes (when progress holds for all of them). However, since our framework supports arbitrary forms of effectful computations, at times it will be necessary to exclude some codes from the separator. As a simple example, consider a set of codes that contains $\codef{fail}$, defined such that $\codef{fail} \app c = \emptyset$ for every $c\in\mca$. 
While we want to allow codes such as $\codef{fail}$ to be used to define and realize propositions, for consistency, $\codef{fail}$ cannot serve as valid evidence for entailment. 
A more subtle example of a separator for CPS continuations is given in~\Cref{exm:cont}.

In summary, defining realizability semantics from an MCA requires extra structure, captured by the following notion of a monadic core.
\begin{definition}[Monadic Core]
    A \emph{monadic core} is a tuple $\effshell:
    =\left(\mca, \Omega, \postmod, \separator\right)$, where $\mca$ is an MCA  over a $\mathbf{Set}$ monad $M$,   $\postmod$ is an  $M$-modality over $\mca$ and  a complete Heyting prealgebra $\left(\Omega, \leq\right)$, and  $\separator$ is a separator over them.
\end{definition}

\subsubsection{From Monadic Cores to Evidenced Frames}
\label{sec:mcaef}
The next theorem   demonstrates how one can construct evidenced
frames from a monadic core, i.e., an MCA and an associated $M$-modality.
The propositions are taken to be functions from the MCA to the complete Heyting
prealgebra underlying the $M$-modality. 
As explained, this is so that propositions are given an algebraic structure generalizing their usual definition as sets of codes. 
Evidence are elements of the
separator (rather than arbitrary codes), and conceptually the evidence relation holds when the separator maps realizers for the input proposition to
computations that, after they terminate, yield realizers for the output proposition. 

\begin{restatable}[Evidenced Frame over Monadic Core \coqdoc{MCA_EF}]{theorem}{MCAEF}\label{thm:MCAtoEF}
Let $\effshell=\left(\mca, \Omega,\postmod, \separator\right)$ be a monadic core. 
The triple $\left( \Omega^\mca, \separator, \evrel{\cdot}{\cdot}{\cdot} \right)$ forms an evidenced frame, where \negspace
\[ \evrel{\predA}{e}{\predB} \defeq \evexpand{\predA}{e}{\predB}.\]
\end{restatable}

 \begin{proof}[Proof Sketch.]
 We define the logical and program constructs, while the proofs that they satisfy the required properties can be found in~\cite{Coqproofs}. 
Let $\code{\codef{p_{1}}} \defeq \encode{1}{0}$ and $\code{\codef{p_{2}}} \defeq \encode{1}{1}$.
\begin{description}
    \item[\textit{Reflexivity}:]
    Take $\eid \defeq \encode{0}{0}$. 
    \item[\textit{Transitivity}:] 
    Take $\ecomp{e_1}{e_2} \defeq \encode{0}{e_{2} \bullet \left(e_{1} \bullet 0 \right)}$. 
    \item[\textit{Top}:] Take $\top \defeq \lambda e.\boldsymbol{1}$ and $\etrue \defeq \eid$.
    \item[\textit{Conjunction}:] Take $\left(\phi_1 \emeet \phi_2 \right) \left(e\right) \defeq \after{c_{1}}{e \app \codepA}{\phi_{1}\left(c_{1}\right)} \meet \after{c_{2}}{e \app \codepB}{\phi_{2}\left(c_{2}\right)}$, and 
    $     \epair{e_1}{e_2} \defeq  \encode{1}{1 \bullet \left(e_{1} \bullet 0\right) \bullet \left(e_{2} \bullet 0\right)} ,
    \efst \defeq  \encode{0}{0 \bullet \codepA} ,
     \esnd \defeq  \encode{0}{0 \bullet \codepB}.$
    \item[\textit{Universal Implication}:] Take 
    $
    \phi \imp \vec{\phi} \left(e\right)  \defeq  \infimum_{\phi \in \vec{\phi}}\infimum_{c \in \mca}\left(\phi\left(c\right) \sqsupset \after{r}{e \app c}{\phi\left(r\right)}\right)
    $, and $
    \elambda{e} \defeq  \encode{1}{e \bullet \left(\encode{2}{2 \bullet 0 \bullet 1} \bullet 0 \bullet 1\right)} , 
    \eeval \defeq  \encode{0}{\eid \bullet \left( 0 \bullet \codepA\right) \bullet \left( 0 \bullet \codepB\right)}.$
\qedhere
    \end{description}
\end{proof}

The above theorem, together with the $\UFam$ construction~\cite[Def. V.4]{CohMiqTat21} that constructs a realizability tripos from an evidenced frame, thus provides the following realizability semantics for MCAs.
As for the standard PCA-based tripos, a set $I$ is mapped to a family of propositions indexed by $I$, which is given a structure of Heyting prealgebra by considering the existence of a uniform realizer (i.e. compatible with any $i\in I$) to witness the preordering relation. As in the evidenced frame, the M-modality is used to handle computations instead of only values. Morphisms are simply mapped to reindexing functions along them.
\begin{corollary}\label{cor:tripos}
Let $\effshell=\left(\mca, \Omega, \postmod, \separator\right)$ be a monadic core.
Then the following functor $\Trip$ from $\Set^\op$ to the category of Heyting prealgebras is a tripos.
\[
\begin{array}{r@{~~}c@{~~}l}
\Trip(I) &\defeq& \left( (\Omega^\mca)^I, \leq_I\right)
\qquad\quad
\Trip(f)(\varphi) \defeq j \mapsto \varphi(f(j))\\
\varphi \leq_I \psi &\defeq& \exists e\in\separator. \forall i\in I.\forall c\in\mca. \varphi(i)(c) \leq 
\after{r}{e \app c}{\psi(i)\left(r\right)}
\end{array}\]
\end{corollary}

We can now verify that the use of a separator indeed ensures the consistency of the induced evidenced
frame, as shown in the following theorem. 
\begin{restatable}[\coqdoc{agreement}]{theorem}{sepconsistent}
\label{thm:sep_consistent}
Let $\effshell=\left(\mca, \Omega, \postmod, \separator\right)$ be a monadic core.
Then the induced evidenced frame
has an evidence $e \in \mca$ such that $\top \evidence{e} \bot $ iff $\boldsymbol{1} \leq \boldsymbol{0}$.
\end{restatable}

 With that in mind, an $M$-modality over a non-trivial $\Omega$ is called \emph{consistent} when it has a separator.
For example, $\after{x}{m}{\pred\left(x\right)} \defeq \infimum_{x \in m} \pred\left(x\right)$ is not a consistent $M$-modality 
because 
$\cloop \in \separator$ (due to combinatory completeness), and so:
$ \boldsymbol{0} \geq \after{x}{\cloop \app c}{\boldsymbol{0}} = \infimum_{x \in \emptyset} \boldsymbol{0} = \boldsymbol{1} $,  for any code $c \in \separator$.
Therefore,  $\boldsymbol{1} \leq \boldsymbol{0}$, which only holds in the trivial complete Heyting prealgebra, where all elements are equivalent.


\subsubsection{Realizability Examples}
This section illustrates the utility of the MCA framework by providing a few examples of how natural realizability models can be obtained via MCAs.
We first show that the standard realizability tripos can be obtained from the PCA-based monadic core (cf.~\Cref{pca_mca}).

\begin{example}[PCAs \coqdoc{partiality_monad}]
    \Cref{pca_mca} shows that PCAs correspond to the sub-singleton monad 
    $\monad A = \{S\subseteq A \mid \forall x_{1},x_{2}\in S. x_{1} = x_{2} \}$.
   Following the standard intuitions of realizability models based on PCAs, when applying an evidence of $\predA \xle e \predB$ to a code $c$ such that $\predA(c)$,  
   the $M$-modality $\after{x}{e\app c}{\predB}$ should express that the computation $e\app c$ returns a valid code for $\predB$, that is, that there exists such a valid code in the corresponding sub-singleton.
   To that end, we take 
   $\Omega = \P\left(\singleton\right)$ (generalizing the Boolean algebra  $\{0,1\}$ making it compatible with intuitionistic meta-theory), and $\after{x}{m}{\pred(x)} \defeq \supremum_{x\in m} \pred(x)$,  for which the whole set $\mca$ defines as usual a valid separator.
   It is then easy to verify that the evidenced frame obtained from \Cref{thm:MCAtoEF} is the expected one described earlier~\cite{CohMiqTat21}.

   \lcnote{add RCA and SCA}
   

\end{example}

Next, to further demonstrate the uniformity and utility of our framework, we go back to a couple of our MCA instances from~\Cref{sec:examples}. We equip each instance with a corresponding $M$-modality and demonstrate that this indeed recovers the expected model.

\begin{example}[RCA \coqdoc{powerset_monad}]\label{ex:relational}
As described in~\Cref{RCA}, relational realizability use the powerset monad $\monad A = \P \left(A\right)$.
PCAs are then a special case of RCAs, as the sub-singleton is a special case of the powerset monad.
The standard modalities for RCAs are the M-modalities of angelic and demonic nondeterminism, both using $\Omega = \P\left(\singleton\right)$ as in  PCAs.
For angelic nondeterminism, the modality is the same as the one used for PCAs, $\after{x}{m}{\pred(x)} \defeq \supremum_{x\in m} \pred(x)$, and $\mca$ is always a separator.
However, for demonic nondeterminism, we take the infimum rather than the supremum.
To allow for a separator, the definition of a modality
has to conjoin some ``termination'' predicate $m \Downarrow$, that implies progress for the elements of the designated separator, yielding the definition: $\after{x}{m}{\pred(x)} \defeq m\Downarrow \wedge \infimum_{x\in m} \pred(x)$.
\end{example}


\begin{example}[SCA \coqdoc{state_monad}]
As in~\Cref{SCA}, for $\Sigma$ a preordered set of states, SCAs use the increasing (powerset) state monad, $\monad A = \left\{ m:\Sigma\rightarrow \P\left(\Sigma\times A\right)\
\! \mid \! \forall\sigma_{0}\in\Sigma, \left(\sigma_{1},x\right)\in m\left(\sigma_{0}\right).\sigma_{0}\leq\sigma_{1}\right\}
$.
To account for state, predicates have an extra component in $\Sigma$, so $\Omega = \Sigma \rightarrow \P\left(\singleton\right)$, and they are restricted so that they must be ``future-stable'', i.e. upward closed with respect to states.
As in RCAs, there are angelic and demonic modalities: 
the angelic is $\left(\after{x}{m}{\pred(x)}\right)^{\sigma} \defeq \infimum_{\sigma' \geq \sigma} \supremum_{\left(x,\sigma''\right)\in m\left(\sigma'\right)} \left(\pred(x)\right)^{\sigma''}$, and 
the demonic again uses a ``termination'' predicate and is given by:  
$\left(\after{x}{m}{\pred(x)}\right)^{\sigma} \defeq \infimum_{\sigma' \geq \sigma} .  m\left(\sigma'\right) \Downarrow \wedge \infimum_{\left(x,\sigma''\right)\in m\left(\sigma'\right)} \left(\pred(x)\right)^{\sigma''}$.

\end{example}

%


\begin{example}[Continuations \coqdoc{continuation_monad}]\label{exm:cont}
As observed in~\Cref{example:CPS}, the continuation monad provides a particular instance
of an MCA replaying a (CbV) CPS translation. Since Krivine classical realizability~\cite{Krivine09} is known to be equivalent to the composition of CPS translation with a standard intuitionistic realizability interpretation~\cite{OlivaStreicher08}, we can easily replay this construction in our setting using a CPSCA based on the continuation monad $\monad A = \left(A \rightarrow R\right) \rightarrow R$.  
Krivine realizability models crucially realies on a parameter $\pole$, the so-called \emph{pole}, which intuitively contains the computations considered as valid; here, this will simply be a subset of $R$. 
Any pole induces an orthogonality relation between elements of $A$ and of $A\to R$: a function $f$ in the latter is said to be orthogonal to $a$, written $f\bot a$ when $f(a)\in\pole$.

Since our algebras follow a call-by-value discipline, realizers should be defined using two layers of orthogonality~\cite{Munch09focalisation,GardelleMiquey23}: given a formula $A$, its interpretation is primitively defined by a set of values $\llbracket A \rrbracket$. Then its set of opponents are the continuations $k$ orthogonal to any $a\in\llbracket A \rrbracket$, while a realizer will be orthogonal to any such continuation. 
Last, as is usual in Krivine realizability, as soon as the pole is non-empty, there is at least one continuation $k$ and $c_a\in\mca$ such that $k(c_a)\in\pole$. Then, the computation $\code{\mathsf{K}_{k}}\app c_a$, which drops the current continuation and applies instead $k$ to $c_a$, would be a realizer of $\bot$. To circumvent this issue, we consider the set PL of \emph{proof-like} codes obtained by combinatorial completeness extended with the code $\code{\mathsf{cc}}$. As in~\cite{Krivine09}, a pole $\pole$ is consistent if for any proof-like term $m$, there exists a continuation $k$ such that $m(k)\notin\pole$. For any such pole, the set PL defines a separator.

Formally, for $\Omega \defeq \P(\singleton)$, a fixed consistent pole $\pole\subset R$ and $\pred : A\to \Omega$,  we define:\negspace
\[\after{x}{m}{\pred\left(x\right)}\defeq \infimum_{k\in A\to R} \left(\Big(\infimum_{a\in A}\left( \pred(a) \haimp k\bot a\right)\Big) \haimp m\bot k\right)\]
This definition, which satisfies the expected axiom of an $M$-modality and admits the set PL as separator, induces via \Cref{thm:MCAtoEF} an evidenced frame corresponding to an indirect-style presentation of a (call-by-value) Krivine realizability model analogous to the one in \cite{GardelleMiquey23}.
\end{example}


\begin{example}[Parameterized \coqdoc{parametric_monad}]

Bauer~\cite{bauer2024countablereals} provides a construction of a parameterized realizability tripos from a ParCA.
The same tripos can be obtained from the MCA representation given in~\Cref{example:bauer} following the construction in~\Cref{cor:tripos}, taking the following $M$-modality (given a set of parameters $\param$):\negspace
\[ \after{x}{m}{\pred\left(x\right)} \defeq \bigcap_{p \in \param} \bigcup_{x \in m\left(p\right)} \pred\left(x\right). \]
The modality extends the one for PCAs by requiring the computation to yield a value in $\phi$ for every possible parameter in $\param$.
Just as in a PCA, here too the set $\mca$ provides a separator. Moreover, for the induced logic to be non-trivial the set of parameters $\param$ has to be non-empty.


\end{example}

\subsection{Connection to Assembly Models}
\label{sec:assemblies}

To further highlight how MCAs naturally generalize PCAs, we here discuss how the construction of assemblies over PCAs, a foundational technique in the study of realizability toposes~\cite{van2008realizability}, seamlessly extends to MCAs.
Assemblies $\Asm_\A$ over a PCA $\A$,  are pairs $(X,\asmmap{X}{\cdot})$ where $\asmmap{X}{\cdot}$ maps any $x\in X$ to a non-empty subsets of $\A$ witnessing $x$'s existence.
The categories of assemblies $\Asm_\A$ are somewhat simpler to handle, and sufficient to model rich intuitionistic systems like the Calculus of Construction \cite{CarFreSce90assemblies}.  In fact, assemblies can be identified as a particular subcategory of the realizability topos, which in itself is a quasi-topos. Moreover, several works studying completions mechanisms within toposes \cite{RobinsonRosolini90completions,Menni2000completions} emphasized that the realizability topos could be recovered as the ex/reg completions of $\Asm_\A$. 

Recently, following a line of work aiming to provide an algebraic counterpart to Krivine realizability using \emph{implicative algebras}, Castro \emph{et al.} extended the usual definition of assemblies to such algebras~\cite{CasMiqKrz23}. While they manage to prove that the resulting assemblies define a quasi-topos as expected, in this context the usual completion mechanism techniques seem to fail and relating assemblies with the corresponding implicative topos remains an open problem to date.
Following the connection established between implicative algebras and evidenced frames in \cite{CohMiqTat21}, we can adapt their construction to define assemblies over any evidenced frame which, combined with \Cref{thm:MCAtoEF} , provide us with a construction of assemblies over any MCA. Note that assemblies over a PCA or an implicative algebra are then recovered by considering the corresponding evidenced frame.

\begin{definition}[Category $\Asm_{\ef}$]
The category of assemblies over an evidenced frame $\ef = \left(\Phi , E , \rightarrow\right)$, is given by:
\negspace
\begin{description}
    \item[Objects :] an \emph{assembly} over $\ef$ is a tuple $X=\left(\underlying{X},  \real{X}\right)$ where $\underlying{X}$ is a set and $\real{X} : \underlying{X} \rightarrow \Phirel$, where $\Phirel=\{\varphi\in\Phi \mid \exists e\in E. \top \xle e \varphi\}$ is the subset of evidenced relations.
    \item[Morphisms :] given two assemblies $X , Y$, a \emph{morphism} $f$ from $X$ to $Y$ is a function $\underlying{f} : \underlying{X} \rightarrow \underlying{Y}$ s.t. there exists an evidence $\tau_f \in E$, s.t for all $x \in \underlying{X}$, 
    $ \real{X} (x) \evidence{\tau_f} \real{Y}(f\left(x\right))$.
That evidence $\tau_f$ is said to ``\emph{track}'' $f$.
\end{description}
\end{definition}

To prove that this indeed defines a category, it suffices to observe that 
the \emph{identity morphism} over assembly $X$ is simply the identity function $\id_{\underlying{X}}$ over $\underlying{X}$ tracked by the reflexivity evidence $e_{\id}$; and that if $f$ is tracked by $\tau_f$ and $g$ is tracked by $\tau_g$, then their composition $f \circ g$ is the composition of $\underlying{f} \circ \underlying{g}$ tracked by the transitivity evidence $\tau_g;\tau_f$.
Studying more in-depth the category $\Asm_\ef$ is out of the scope of this paper, but 
we conjecture that one could follow the development for implicative assemblies in~\cite{CasMiqKrz23} to prove that the $\Asm_{\ef}$ is also finitely (co)complete, locally cartesian closed and possesses a strong object classifier, and defines then a quasi-topos.
Nonetheless, as for the implicative case, it is not clear whether some mechanism analogous to the ex/reg completion could connect it to the topos induced by the evidenced frame.



\section{Conclusion and Future Work}\label{sec:conc}
This paper introduces Monadic Combinatory Algebras (MCAs) as a novel extension of PCAs that encompasses a wide range of computational effects through the use of monads. This new framework addresses the limitations of traditional PCAs, which only support non-termination as a computational effect, by providing a more comprehensive model capable of internalizing effects such as nondeterminism, stateful computations, continuations and oracles.
%
%
We link MCAs to realizability theory (generalizing the role of PCAs in traditional realizability models)
by providing two uniform constructions of realizability models from MCAs, triposes and assemblies, that factor through evidenced frames. 
 Overall, MCAs provide a powerful and flexible framework for internalizing computational effects that opens up new avenues in the study of effectful computations and their algebraic and categorical models. 
 

Future research into the MCA framework presents several directions for exploration. 
A key one is examining how different underlying monads affect the resulting theory, potentially providing novel insights into constructive models of computation. 
Further work is also needed to extend the MCA scope to unaddressed effects, such as probabilistic computation. 
Notably, our MCA structure is based on $\setcat$-monads, whereas probabilistic computation is typically formalized using the Giry monad~\cite{giry2006categorical} in the category $\mathbf{Meas}$ (measurable spaces and functions). Advancing the theory of combinatory objects in Freyd categories 
will enable the exploration of broader computational effects beyond $\setcat$-based monads.
Moreover, a more comprehensive categorical understanding of MCAs, via their connections to structures like monoidal, enriched, and higher categories, is needed.
\lcnote{discuss extending the categorical setting to the modality and cite}

The MCA framework can also be further extended to support more diverse and complex notions of computation such as hybrid effects, where multiple monads are combined to model complex interactions between different computational effects. 
Additionally, the framework can be enriched by incorporating alternative evaluation strategies, such as Call-by-Name, which may uncover new computational behaviors. 



\clearpage

\bibliography{ref}
\appendix
\newpage
\onecolumn
\clearpage

\appendix
\vshort{
\section*{Appendix: Elaborated \Cref{sec:turing} --- Categorical Characterization of MCAs}

\setcounter{lemma}{0}
\renewcommand{\thelemma}{\Alph{section}\arabic{lemma}}
\setcounter{definition}{0}
\renewcommand{\thedefinition}{\Alph{section}\arabic{definition}}

We first recall the formal definition of Freyd categories and some associated categorical components~\cite{POWER2002Premonoidal}. 
\revnote{Some further explanation is needed to motivate the
concepts of binoidal category, etc. It seems like the idea is that
the tensor is separately functorial in each argument, but not
simultaneously functorial. By that token, one might guess that this
is like a semigroup in the monoidal 1-category of (strict)
categories equipped with the so-called "funny" tensor product. If
that is so, would a "premonoidal" category be a monoid object in
(Cat,FunnyTensor)? And what is the significance of various morphisms
being central? Anyway, all this is surely fine, but some explanation
that motivates these things would help make the paper more
accessible to less specialist readers. Ultimately it seems that
these notions are aiming to build in a failure of interchange,
corresponding to the way that effects can't usually re-ordered.}
We start with binoidal categories, which capture the idea of non-commutativity of general effectful computation.
For example, given two computations $m_{1} , m_{2}$ in the state monad (\Cref{SCA}), then $\letin{x_{1}}{m_{1}}{\letin{x_{2}}{m_{2}}{\left(x_{1} , x_{2}\right)}}$ will generally not yield the same computation as $\letin{x_{2}}{m_{2}}{\letin{x_{1}}{m_{1}}{\left(x_{1} , x_{2}\right)}}$, because in the former case $m_{2}$ depends on the state modified by $m_{1}$, while in the latter case $m_{1}$  depends on the state modified by $m_{2}$, so the order in which they are sequenced matter.
Binoidal categories abstract this behavior.

\begin{definition}[Binoidal Category]
    A category $\C$ is a \emph{binoidal category} if :
    \begin{itemize}[leftmargin=0.5cm]
        \item For every  $A , B \in\C$, there is an object $A \otimes B$
        \item For every $A\in\C$, there is a functor $A \rtimes \left(-\right)$, sending morphisms in $\hom{\C}{B_{1}}{B_{2}}$ to  morphisms in $\hom{\C}{A \otimes B_{1}}{A \otimes B_{2}}$
        \item For every object $B \in \C$, a functor $\left(-\right) \ltimes B$ , sending morphisms in $\hom{\C}{A_{1}}{A_{2}}$ to  morphisms in $\hom{\C}{A_{1} \otimes B}{A_{2} \otimes B}$
    \end{itemize}
\end{definition}

A morphism $f : A \rightarrow B$ in a binoidal category $\K$ is \emph{central} if for every other morphism $u : X \rightarrow Y$ in $\K$, $\left(u \ltimes B\right) \circ \left(X \rtimes f\right) = \left(Y \rtimes f\right) \circ \left(u \ltimes A\right)$ and $\left(f \ltimes Y\right) \circ \left(A \rtimes u\right) = \left(B \rtimes u\right) \circ \left(f \ltimes X\right)$.
Given a binoidal category $\K$,  the centre of $\K$ is 
the subcategory of $\K$ consisting of all the objects of $\K$ and the central morphisms.

\begin{definition}[Symmetric Premonoidal Category]
A \emph{symmetric premonoidal category} is a binoidal category equipped with an object $\mathbb{I}$  and  the following central natural isomorphisms:
\begin{description}[font=\normalfont,leftmargin=0.5cm]
    \item[Associator:]~~~~ 
         $\alpha_{A,B,C} 
         : \left(A \otimes B\right) \otimes C \rightarrow A \otimes \left(B \otimes C\right)$ 
    \item[Left and right unitor:]~~~~ 
        $ \lambda_{A} : 
        \mathbb{I} \otimes A \rightarrow A $ and  $ \rho_{A} 
        : A \otimes \mathbb{I} \rightarrow A $
    \item[Swap:]~~~~ 
        $ \sigma_{A,B} : 
       A \otimes B \rightarrow B \otimes A $
\end{description}
The above natural isomorphisms obey the triangle and pentagon coherence laws as their counterparts in a monoidal category.
\end{definition}


\begin{definition}[Freyd Category]
    A Freyd category is a triple $\langle \C, \K, \purefun \rangle$ such that:
    \begin{itemize}[leftmargin=0.5cm]
        \item $\C$ is a cartesian category 
        \item $\K$ is a symmetric premonoidal category with the same objects as $\C$
     \item $\purefun : \C \rightarrow \K$ is an identity on objects functor, strictly preserving symmetric premonoidal structure, whose image lies inside the centre of $\K$.
     \end{itemize}
\end{definition}

To define \emph{combinatory objects} we extend the $\aparr$ morphism of the applicative object to an $n$-ary operator.
The iterated product $A^{n}$ of an object $A$ is recursively defined as:
$A^{0} \defeq \termobj$, $A^{1} \defeq A$, and 
$A^{n+2} \defeq A \times A^{n+1}$.
\agnote{TODO: should also generalize associator}
Then, projections and pairings are similarly generalized.
The association $\alpha$ is generalized to normalize pairs of arbitrary iterated products to a single iterated product, with $\alpha^{0} \defeq \lambda$, $\alpha^{1} \defeq \id$, and $\alpha^{n+2}$ defined as the composition:
\[\begin{tikzcd}
	{A^{k+2} \times A^{m}} & {A \times \left(A^{k+1} \times A^{m}\right)} & {A^{m+k+2}}
	\arrow["{\alpha}", from=1-1, to=1-2]
	\arrow["{A \rtimes \alpha^{k+1}}", from=1-2, to=1-3]
\end{tikzcd}\]

Finally, $\aparr^{n} \in \hom{\K}{\mca^{n+1}}{\mca}$ is defined by: $\aparr^{0} \defeq \rho$,
$\aparr^{1} \defeq \aparr$, and taking  $\aparr^{n+2}$ to be:
\[\begin{tikzcd}
	{\mca \times \left(\mca \times \mca^{n+1}\right)} & {\left(\mca \times \mca\right) \times \mca^{n+1}} & {\mca \times \mca^{n+1}} & \mca
	\arrow["{\alpha^{-1}}", from=1-1, to=1-2]
	\arrow["{\aparr \ltimes \mca^{n+1}}", from=1-2, to=1-3]
	\arrow["{\aparr^{n+1}}", from=1-3, to=1-4]
\end{tikzcd}\]

\begin{definition}[$\mca$-monomials]
Given a applicative object $\mca$, $\mca$-monomials are defined using expressions similar to the ones in \Cref{subsection:PCA}:
$$
\expr {}::={}  i \in \mathbb{N} \mid c \in \hom{\C}{\termobj}{\mca} \mid \expr \bullet \expr \qquad\qquad
E_{n}\left(\mca\right)  {}::={}  \{ \expr \mid \text{all $i$s in $\expr$ are $< n$}\} 
$$

The evaluation function $\interp{-}_{n}  : E_{n}\left(\mca\right) \rightarrow \hom{\K}{\mca^{n}}{\mca}$ relates expressions to $\K$-morphisms:
\begin{align*}
    \interp{i}_{n} & \defeq \purefun \pi_{i+1}^{n}\\
    \interp{c}_{n} & \defeq \purefun \left(c \; \circ \; !\right)\\
    \interp{\expr_{f} \bullet \expr_{a}}_{n} &\defeq \aparr \circ \left(\mca \rtimes \interp{\expr_{a}}_{n} \right) \circ \left(\interp{\expr_{f}}_{n} \ltimes \mca^{n} \right) \circ \purefun \Delta
\end{align*}

An $\mca$-monomial is a morphism $f$ in $\K$ s.t there exists $n \in \mathbb{N}$ and a term $\expr_{f} \in E_{n}\left(\mca\right)$ for which $\interp{\expr_{f}}_{n} = f$.

\end{definition}

\begin{lemma}\label{lemma:eval-subst}
For all $n \in \mathbb{N}$, all $\expr \in E_{n}\left(\mca\right)$, and all $c_{1} , \ldots , c_{n} \in \hom{\C}{\termobj}{\mca}$:
\[ \interp{ \expr \sub{c_{1}} \cdots \sub{c_{n}} }_{0} = \interp{\expr}_{n} \circ \purefun \pairing{ c_{1} , \ldots , c_{n} }_{n} \]
\end{lemma}
\begin{proof}
    By structural induction on $\expr$.
\end{proof}

\begin{proposition}
    Let $\mca$ be an applicative object.
$\mca$ is a combinatory object iff the $\kcomb$ and $\scomb$ morphisms in $\K$ defined below are $\mca$-computable.
$$\begin{array}{l@{\hspace{0.2cm}}l}
    \kcomb : \mca^{2} \rightarrow \mca & 
    \kcomb \defeq \purefun \pi_{1} \\
    \scomb : \mca^{3} \rightarrow \mca &
    \scomb \defeq \aparr \circ \left(\mca \rtimes \aparr\right) \circ \left(\aparr \ltimes \mca^{2} \right) \circ \purefun \pair{\pair{\pi^{3}_{1}}{\pi^{3}_{3}}}{\pair{\pi^{3}_{2}}{\pi^{3}_{3}}}
\end{array}$$
\end{proposition}

\begin{proof}
    Let $\mca$ be an applicative object.
    Assuming all positive $\mca$-monomials are $\mca$-computable, we consider the expressions $\expr_{\kcomb} \in E_{2}\left(\mca\right)$ and $\expr_{\scomb} \in E_{3}\left(\mca\right)$:
$$
        \expr_{\kcomb} \defeq 0 \quad  \quad
        \expr_{\scomb} \defeq \left(0 \bullet 2\right) \bullet \left(1 \bullet 2\right)
$$
    Since $\interp{\expr_{\kcomb}}_{2} = \kcomb$ and $\interp{\expr_{\scomb}}_{3} = \scomb$, then $\kcomb$ and $\scomb$ are $\mca$-monomials, and thus $\mca$-computable.

    Conversely, if $\kcomb$ and $\scomb$ are $\mca$-computable, let $\kcode$ and $\scode$ be their corresponding codes (with their associated codes of partial applications).
    For every term $\expr \in E_{n+1}\left(\mca\right)$, we define a code $\encode{n}{\expr} \in \hom{\C}{\termobj}{\mca}$ using the following bracket abstraction algorithm.
    First, for every $\expr \in E_{n+1}\left(\mca\right)$, we define a closed abstraction expression $\absterm{n}{\expr} \in E_{0}\left(\mca\right)$ as follows:
$$
\begin{array}{l@{\qquad}|>{\qquad}l}
    \absterm{n}{0} \defeq K_{n} &
    \absterm{n+1}{j+1} \defeq \kcode \bullet \absterm{n}{j}\\
    \absterm{n}{c} \defeq K_{n+1} \bullet c &
    \absterm{n}{\expr_{1} \bullet \expr_{2}} \defeq S_{n+1} \bullet \absterm{n}{\expr_{1}} \bullet \absterm{n}{\expr_{2}}
\end{array}$$
The definition uses $K_{n}$ and $S_{n}$, which are the $n$-ary $K$ and $S$ combinators~\cite{Goldberg-Mayer:2015},  
 where 
  $  B  \defeq \scode \bullet \left( \kcode \bullet \scode \right) \bullet \kcode$:
$$
\begin{array}{l@{\qquad\quad}l@{\qquad\quad}l}
     K_{0} \defeq \scode \bullet \kcode \bullet \kcode 
     & K_{1} \defeq \kcode & 
     K_{n+2} \defeq B \bullet \kcode \bullet K_{n+1}
     \\
     S_{0} \defeq \scode \bullet \kcode \bullet \kcode & 
     S_{1} \defeq \scode &
    S_{n+2} \defeq B \bullet \scode \bullet \left( B \bullet S_{n+1}\right)
\end{array}
$$
Since $\kcomb$ and $\scomb$ are $\mca$-computable, the axioms ensure that for every $\expr \in E_{n+1}\left(\mca\right)$ there is a code $\encode{n}{\expr} \in \hom{\C}{\termobj}{\mca}$ such that $\interp{\absterm{n}{\expr}}_{0} = \purefun \encode{n}{\expr}$, and for all $c_{1} , \ldots , c_{k} \in \hom{\C}{\termobj}{\mca}$ (where $k<n$) the following hold:
\begin{align}  
        \interp{\encode{n-k}{ \expr\sub{c_{1}}\cdots\sub{c_{k}}} \bullet 0 \bullet \cdots \bullet n-k-1}_{n-k} &= \interp{ \expr\sub{c_{1}}\cdots\sub{c_{k}} }_{n-k}
        \\
            \interp{\encode{n}{\expr} \bullet c_{1} \bullet \cdots \bullet c_{k}}_{0} &= \interp{\encode{n-k}{ \expr\sub{c_{1}}\cdots\sub{c_{k}}}}_{0}
\end{align}
Using \Cref{lemma:eval-subst}, the above equations correspond to the square and triangle diagrams of \Cref{def:computable}.
Hence we take $\code{f\left(c_{1},\ldots,c_{k}\right)}$ to be a $\encode{n-k}{\expr_{f}\sub{c_{1}}\cdots\sub{c_{k}}}$, for which it is straightforward to verify the combinatory object conditions.
\end{proof}

Next we relate the set-based MCAs to the abstract categorical notion of combinatory objects.
For this, we recall that the Kleisli category of every strong monad is a premonoidal category~\cite{power1997premonoidal}. 
Here, we work in $\setcat$, where every monad is strong.
Together, this entails that for every monad, the Kleisli category gives rise to a Freyd category.
Since the details of the construction are relevant to our result, we here provide them below.

\begin{lemma}
Given a $\setcat$ monad $\monad$, the triple $\langle\setcat , \setcat_{\monad}, \purefun_{\monad}\rangle$
forms a Freyd category, where $\setcat_{\monad}$ is the Kleisli category of $\monad$ over $\setcat$,  and $\purefun_{\monad} : \setcat \rightarrow \setcat_{\monad}$ is the canonical functor defined by $\purefun_{\monad}\left(f\right) = \eta \circ f$.
\end{lemma}
\begin{proof}
Given a monad $\monad$, the constructs making $\left(\setcat , \setcat_{\monad}, \purefun_{\monad} \right)$ a Freyd category are:
\begin{itemize}[leftmargin=0.5cm]
    \item The cartesian structure of $\setcat$ is given by the cartesian product of sets $\times$
    \item For 
    every function $f : A_{1} \rightarrow \monad A_{2}$: 
    $ \left(f \ltimes B \right)\left(x_{1},y\right) = \letin{x_{2}}{f\left(x_{1}\right)}{\return{ \left(x_{2}, y\right) } } $
    \item For 
    every function $f : B_{1} \rightarrow \monad B_{2}$: 
    $\left(A \rtimes f\right)\left(x,y_{1}\right) = \letin{y_{2}}{f\left(y_{1}\right)}{\return{ \left(x, y_{2}\right) } } $
    \item The associator, left unitor, right unitor, and swap  are given by composing $\eta$ on the associator, left unitor, right unitor, and swap of the underlying cartesian structure of $\setcat$.
\end{itemize}
The Freyd requirements for this construction are easily verified.
\end{proof}


Since the application morphism of an applicative system is a morphism in $\K$, which is the category of computations in a Freyd category, it corresponds to the Kleisli application of an MAS.

\begin{lemma}\label{lem:MAS}
Let $\monad$ be a $\setcat$ monad.
    $\mca$ is a MAS over $\monad$ if and only if it is an applicative object in $\setcat_{\monad}$.
\end{lemma}
\begin{proof}
    The application morphism of the applicative object is the application Kleisli function of the MAS: 
    $\aparr\left(c_{1} , c_{2}\right) = c_{1} \app c_{2}$. 
\end{proof}

\combObj*

\begin{proof}
We use the $\scomb$ and $\kcomb$ formulation of an MCA given in~\Cref{prop:SK} to prove the equivalence to $\scomb$ and $\kcomb$ being $\mca$-computable. 
Using the same set of codes $\scode , \kcode , \scodeA{c_{1}} , \scodeB{c_{1}}{ c_{2}}$, and $\kcodeA{c_{1}}$ for any two codes $c_{1}$, $c_{2}$, it is easy to verify that the equations in~\Cref{prop:SK} are equivalent to those in~\Cref{def:computable}.
\end{proof}

The categorical definition of a PCA-object in Turing categories relies on a specific notion of products, which, in turn, relies on a general notion of restrictions~\cite{COCKETT2008}.\label{sec:restriction}
A restriction structure is a convenient way of handling partiality in category theory.
Roughly speaking, considering morphisms as partial maps, a restriction assigns to each partial map $f$ a partial identity map $\restrict{f}$ with the same domain as $f$.
To relate to that definition, we here show that cartesian restriction categories are a special case of Freyd categories.  

\begin{definition}[Restriction structure]
A \emph{restriction structure} on a category $\C$ assigns to every morphism $f : A \rightarrow B$ a morphism $\restrict{f} : A \rightarrow A$  s.t. the following hold:
\begin{itemize}[leftmargin=0.5cm]
    \item For 
    $f : A \rightarrow B$:
    $ f \circ \restrict{f} = f $
    \item For  
    $f : A \rightarrow B_{1}$ and $g : A \rightarrow B_{2}$:
     $\restrict{f} \circ \restrict{g} = \restrict{g} \circ \restrict{f}$ and $ \restrict{f \circ \restrict{g}} = \restrict{f} \circ \restrict{g}$ 
    \item For 
    $f : B \rightarrow C$ and $g : A \rightarrow B$:
     $\restrict{f} \circ g = g \circ \restrict{f \circ g} $
\end{itemize}
\end{definition}

A morphism $f$ is \emph{total} when $\restrict{f} = \id$.
The total morphsims in a category $\C$ form a subcategory of $\C$, called $\text{Tot}\left(\C\right)$.
A restriction in $\C$ induces a partial order on the hom-sets of $\C$ as follows 
$ f \leq g  \iff f = g \circ \restrict{f}$.
%
A \emph{cartesian restriction category} is a category with a restriction structure along with restriction cartesian products and restriction terminal objects which are  weaker variants of the standard notions that still form a monoidal category~\cite{cockett2007restriction}.
In a cartesian restriction category $\C$, the subcategory $\text{Tot}\left(\C\right)$ is a cartesian subcategory of $\C$.

\begin{proposition}
    Cartesian restriction categories induce Freyd categories. 
\end{proposition}

\begin{proof}
    Let $\C$ be a cartesian restriction categories. The induced Freyd category is defined by $\langle \text{Tot}\left(\C\right) , \C  , \purefun \rangle$
    where $\purefun$ is the inclusion functor of $\text{Tot}\left(\C\right)$ into $\C$.
It is straightforward to verify that it indeed forms a Freyd category.
\end{proof} 

The definition of a PCA-object in a cartesian restriction category~\cite{COCKETT2008}, using S and K combinators, is the exact counterpart of our definition of a combinatory object in a Freyd category, by merely replacing morphisms of a Freyd category with their counterparts in a cartesian restriction category.
In fact, the additional requirement needed in that formulation, namely that the code morphism has to be total, is subsumed by taking it from the underlying cartesian category  through the functor given in the definition of the Freyd category.

\begin{proposition}
In Freyd categories induced by cartesian restriction categories, 
combinatory objects are PCA-objects, as defined in~\cite{COCKETT2008}.
\end{proposition}

}
\vlong{
\appendix
\section{Proof of \Cref{pca_mca}}
\begin{proof}
A PCA is obtained from an MCA by instantiating the monad with the sub-singleton monad, i.e.  $M A$ is the set of subsets of $A$ in which all elements are equal.
\[
  \monad A = \{S\subseteq A \mid \forall x_{1},x_{2}\in S. x_{1} = x_{2} \}
 \qquad\qquad
\eta_A\left(x\right) = \{x\}
 \qquad\qquad
 \mu_A\left(\mathbb{X}\right) = \bigcup_{X\in\mathbb{X}} X
\]
The sub-singleton monad, of sets with at most one element, allows for the interpretation of deterministic binary relations as functions, where the codomain is in the monad.
Given a deterministic relation $R \subseteq A \times B$, it is equivalent to a function $\widehat{R} : A \rightarrow \{S\subseteq B \mid \forall y_{1},y_{2}\in S. y_{1} = y_{2} \}$ where $\widehat{R}(a)= \{b\}$ if $R(a,b)$ and $\widehat{R}(a)=\emptyset$ otherwise.
Because the evaluation relation is determinstic, using the sub-singleton monad to define an MCA, the MCA laws become the laws of a PCA.
That is,  instead of using $c_{f} \app c_{a} \downarrow c$ and $e \downarrow c$ in a PCA, we use $c_{f} \app c_{a} = \left\{ c \right\}$ and $\nu\left(e\right) = \left\{ c \right\}$ in the MCA, respectively. We note that, when working in a classical metatheory, PCAs can also be obtained by instantiating MCAs with the maybe monad.
\end{proof}

\section{Proof of~\Cref{prop:SK}}
\begin{proof}
    Let $\mca$ be a monadic applicative structure. 
    If $\mca$ is an MCA then  the $\scode$ and~$\kcode$ combinators are the codes $\encode{2}{(0 \app 2) \app (1 \app 2)}$ and~$\encode{1}{0}$ modeling the $\lambda$-calculus terms $\lambda x.\lambda y.\lambda z.\,(x~z)~(y~z)$ and~$\lambda x.\lambda y.x$, respectively.
    
    Conversely, if $\mca$ has $\scode$ and $\kcode$ combinators satisfying the axioms, then for every term $e \in E_{n+1}\left(\mca\right)$, we can define a code $\encode{n}{e}$ using the following bracket abstraction algorithm.
    First, we define a closed abstraction term $\absterm{n}{e} \in E_{0}\left(\mca\right)$ for every $n \in \mathbb{N}$ and any $e \in E_{n+1}\left(\mca\right)$ as follows:
$$
\begin{array}{l@{\qquad}|>{\qquad}l}
    \absterm{n}{0} \defeq K_{n} &
    \absterm{n+1}{j+1} \defeq \kcode \bullet \absterm{n}{j}\\
    \absterm{n}{c} \defeq K_{n+1} \bullet c &
    \absterm{n}{e_{1} \bullet e_{2}} \defeq S_{n+1} \bullet \absterm{n}{e_{1}} \bullet \absterm{n}{e_{2}}
\end{array}$$
The definition uses $K_{n}$ and $S_{n}$, which are the $n$-ary $K$ and $S$ combinators~\cite{Goldberg-Mayer:2015},  
 where 
  $  B  \defeq \scode \bullet \left( \kcode \bullet \scode \right) \bullet \kcode$:
$$
\begin{array}{l@{\qquad\quad}l@{\qquad\quad}l}
     K_{0} \defeq \scode \bullet \kcode \bullet \kcode 
     & K_{1} \defeq \kcode & 
     K_{n+2} \defeq B \bullet \kcode \bullet K_{n+1}
     \\
     S_{0} \defeq \scode \bullet \kcode \bullet \kcode & 
     S_{1} \defeq \scode &
    S_{n+2} \defeq B \bullet \scode \bullet \left( B \bullet S_{n+1}\right)
\end{array}
$$
The axioms ensure there is a code $c_{n , e} \in \mca$ for any $n \in \mathbb{N}$ and any $e \in E_{n+1}\left(\mca\right)$ such that $\nu\left(\absterm{n}{e}\right) = \return{c_{n , e}}$, hence we take $\encode{n}{e}$ to be a $c_{n , e}$, for which it is straightforward to verify the MCA conditions.
\end{proof}

\clearpage
\setcounter{lemma}{0}
\renewcommand{\thelemma}{\Alph{section}\arabic{lemma}}

\section{Elaborated \Cref{sec:turing} --- Categorical Characterization of MCAs}

\clearpage
\section{Elaborated~\Cref{sec:mcaef}: From Effect Shells to Evidenced Frames}

We start by providing a series of useful lemmas. 

\begin{restatable}{lemma}{extmon}\label{lemma:after-mono}
    For all $A$, $\predA , \predB : A \rightarrow \Omega$, and $m \in M\left(A\right)$: 
$$\begin{array}{c}       \dfrac{
\predA\left(x\right) \leq \predB\left(x\right)}{\after{x}{m}{\predA\left(x\right)}\leq\after{x}{m}{\predB\left(x\right)}}
    \end{array}$$
\end{restatable}

    \begin{proof}
    Assume for all $ x \in A $, $ \predA\left(x\right) \leq \predB\left(x\right)$. By currying, we get:
    \[\boldsymbol{1} \leq \predA\left(x\right) \haimp \predB\left(x\right) \]
    Since this holds for all $x \in A$, then:
    \begin{align*}
        & \boldsymbol{1} \leq \infimum_{x \in A}\left(\predA\left(x\right) \haimp \predB\left(x\right)\right)
        \leql{\eqref{eq:after-inter-mono}} \after{x}{m}{\predA\left(x\right)} \haimp \after{x}{m}{\predB\left(x\right)}
    \end{align*}
    From which, by uncurrying, we obtain:
    \[ \after{x}{m}{\predA\left(x\right)} \leq \after{x}{m}{\predB\left(x\right)}. \]
    \qedhere
\end{proof}

\begin{lemma}\label{lemma:after-imp}\label{lemma:after-conj}
    For all $A$, $\pred : A \rightarrow \Omega$, $\theta \in \Omega$, and $m \in M\left(A\right)$:
    \begin{enumerate}
        \item 
$    \after{x}{m}{\left(\theta \haimp \pred\left(x\right)\right)} \leq \theta \haimp \after{x}{m}{\pred\left(x\right)}$
    \item 
$    \theta \meet \after{x}{m}{\pred\left(x\right)} \leq \after{x}{m}{\left(\theta \meet \pred\left(x\right)\right )}$
    \end{enumerate}
\end{lemma}
\begin{proof} ~
\begin{enumerate}
    \item 
    By reflexivity, we have:
    \[ \forall x \in A . \theta \haimp \pred\left(x\right) \leq \theta \haimp \pred\left(x\right) \]
    by uncurrying $\theta$ and then currying $\theta \haimp \pred\left(x\right)$, we get:
    \[ \forall x \in A . \theta \leq \left(\theta \haimp \pred\left(x\right)\right) \haimp  \pred\left(x\right) \]
    so by the infimum property we get:
    \begin{align*}
        \theta  \leq \infimum_{x \in A} \left(\left(\theta \haimp \pred\left(x\right)\right) \haimp \pred\left(x\right) \right)
        \leql{\eqref{eq:after-inter-mono}} \after{x}{m}{\left(\theta \haimp \pred\left(x\right)\right)} \haimp \after{x}{m}{ \pred\left(x\right)}
    \end{align*}
    Now, by uncurrying $\after{x}{m}{\left(\theta \haimp \pred\left(x\right)\right)}$ and currying $\theta$ we get:
    \[ \after{x}{m}{\left(\theta \haimp \pred\left(x\right)\right)} \leq \theta \haimp \after{x}{m}{ \pred\left(x\right)} \]
    \item 
    By reflexivity, we have:
    \[ \forall x \in A . \theta \meet \pred\left(x\right) \leq \theta \meet \pred\left(x\right) \]
    by currying $\pred\left(x\right)$ we get:
    \[ \forall x \in A . \theta \leq \pred\left(x\right) \haimp \left(\theta \meet \pred\left(x\right)\right) \]
    so by the infimum property we get:
    \begin{align*}
        \theta  \leq \infimum_{x \in A} \left( \pred\left(x\right) \haimp \left(\theta \meet \pred\left(x\right)\right) \right)
        \leql{\eqref{eq:after-inter-mono}} \after{x}{m}{\pred\left(x\right)} \haimp \after{x}{m}{\left(\theta \meet \pred\left(x\right)\right)}
    \end{align*}
    Now, by uncurrying $\after{x}{m}{\pred\left(x\right)}$ we get:
    \[ \theta \meet \after{x}{m}{\pred\left(x\right)} \leq \after{x}{m}{\left( \theta \meet \pred\left(x\right) \right)} \]
    \qedhere
\end{enumerate}
\end{proof}


\MCAEF*

\begin{proof}
We define the logical and program components and show that they satisfy the required properties in each case.
We use the abbreviations: $\code{\codef{p_{1}}} \defeq \encode{1}{0}$ and $\code{\codef{p_{2}}} \defeq \encode{1}{1}$.
\begin{description}[leftmargin=*]
    \item[Reflexivity:]
    Take $\eid \defeq \encode{0}{0}$. 
    We  need to show that $$ \evexpand{\pred}{\codeid}{\pred}.$$
    For $c \in \mca$:
    \[ \pred\left(c\right) 
         \leq  \after{r}{\return{c}}{\pred\left(r\right)} = \after{r}{\encode{0}{0} \app c}{\pred\left(r\right)} 
        = \after{r}{\codeid \app c}{\pred\left(r\right)} \]
    
    \item[Transitivity:] 
    Take $\ecomp{e_1}{e_2} \defeq \encode{0}{e_{2} \bullet \left(e_{1} \bullet 0 \right)}$. 
    We need to show: 
        \[ \dfrac{
        \begin{array}{c}
            \evexpand{\predA}{e_{s}}{\predB}\\
            \evexpand{\predB}{e_{t}}{\predC}
        \end{array}}
    {\evexpand{\predA}{\code{e_{s} ; e_{t}}}{\predC}} \]
    Given the assumptions, for $c \in \mca$:
    \[
    \begin{array}{cl}
         & \predA\left(c\right)\\
         \leq & \after{s}{e_{s} \app c}{\predB\left(s\right)}\\
         \leql{\Cref{lemma:after-mono}} & \after{s}{e_{s} \app c}{\after{r}{e_{t} \app s}{\predC\left(r\right)}}\\
         \leql{\text{A.Bind}} & \after{r}{\left(\letin{s}{e_{s} \app c}{e_{t} \app s}\right)}{\predC\left(r\right)}\\
         = & \after{r}{\encode{0}{e_{t} \bullet \left(e_{s} \bullet 0\right)} \app c}{\predC\left(r\right)} 
     = \after{r}{\code{e_{s} ; e_{t}} \app c}{\predC\left(r\right)}
    \end{array}
    \]
    
    \item[Top:] Take $\top \defeq \lambda e.\boldsymbol{1}$ and $\etrue \defeq \eid$. 
    We need to show: 
    \[ \evexpand{\pred}{\codeid}{\top} \]
    For $c \in \mca$:
    \begin{align*}
        \pred\left(c\right) \leq \boldsymbol{1} \leql{\text{A.Return}} \after{r}{\return{c}}{\boldsymbol{1}}
        = \after{r}{\return{c}}{\top\left(r\right)} = \after{r}{\codeid \app c}{\top\left(r\right)}
    \end{align*}

    \item[Conjunction:] Take:
    $$\begin{array}{r@{\hspace{0.05cm}}c@{\hspace{0.05cm}}l}
    \left(\phi_1 \emeet \phi_2 \right) \left(e\right) &\defeq& \after{c_{1}}{e \app \codepA}{\phi_{1}\left(c_{1}\right)} \meet \after{c_{2}}{e \app \codepB}{\phi_{2}\left(c_{2}\right)} \\  
    \epair{e_1}{e_2} &\defeq & \encode{1}{1 \bullet \left(e_{1} \bullet 0\right) \bullet \left(e_{2} \bullet 0\right)} \\
    \efst &\defeq & \encode{0}{0 \bullet \codepA} \\
     \esnd &\defeq & \encode{0}{0 \bullet \codepB}
    \end{array}$$
    For \emph{intro}, 
    we need to show:
        \[ \dfrac{
        \begin{array}{c}
            \evexpand{\pred}{e_{1}}{\predA}\\
            \evexpand{\pred}{e_{2}}{\predB}
        \end{array}}
    {\evexpand{\pred}{\codepair{e_{1}}{e_{2}}}{\left(\predA \emeet \predB\right)}} \]
    By combining the assumptions, we get, for any $c \in \mca$:
    \[
    \begin{array}{cl}
         & \pred\left(c\right)\\
         \leq & \after{c_{2}}{e_{2} \app c}{\predB\left(c_{2}\right)} \meet  \after{c_{1}}{e_{1} \app c}{\predA\left(c_{1}\right)}\\
         \leql{\Cref{lemma:after-conj}} & \after{c_{1}}{e_{1} \app c}{\left(\after{c_{2}}{e_{2} \app c}{\predB\left(c_{2}\right)} \meet \predA\left(c_{1}\right)\right)}\\
         \leql{\Cref{lemma:after-mono}} & \after{c_{1}}{e_{1} \app c}{\left(\predA\left(c_{1}\right) \meet \after{c_{2}}{e_{2} \app c}{\predB\left(c_{2}\right)}\right)}\\
         \leql{\Cref{lemma:after-mono} + \Cref{lemma:after-conj}} & \after{c_{1}}{e_{1} \app c}{\left(\after{c_{2}}{e_{2} \app c}{\left(\predA\left(c_{1}\right) \meet \predB\left(c_{2}\right)\right)}\right)}
    \end{array}
    \]
    
    Now, by using \Cref{lemma:after-mono} twice over each of the projections of $\meet$, we get:
    \begin{enumerate}
        \item $ \pred\left(c\right) \leq \after{c_{1}}{e_{1} \app c}{\after{c_{2}}{e_{2} \app c}{\predA\left(c_{1}\right)}} $

        \item $ \pred\left(c\right) \leq \after{c_{1}}{e_{1} \app c}{\after{c_{2}}{e_{2} \app c}{\predB\left(c_{2}\right)}} $
    \end{enumerate}

    Using (1) we obtain: 
    \[
    \begin{array}{cl}
        & \pred\left(c\right)\\
        \leq & \after{c_{1}}{e_{1} \app c}{\after{c_{2}}{e_{2} \app c}{\predA\left(c_{1}\right)}}\\
        \leql{\Cref{lemma:after-mono} +             \text{A.Return}} & \after{c_{1}}{e_{1} \app c}{\after{c_{2}}{e_{2} \app c}{ \after{r_{1}}{\return{c_{1}}}{\predA\left(r_{1}\right)} }}\\
        \leql{\Cref{lemma:after-mono} + \text{A.Bind}} & \after{c_{1}}{e_{1} \app c}{ \after{r_{1}}{\left( \letin{c_{2}}{e_{2} \app c}{\return{c_{1}}} \right)}{ \predA\left(r_{1}\right) } }\\
        \leql{\text{A.Bind}}  & \after{r_{1}}{\left(\letin{c_{1}}{e_{1} \app c}{ \letin{c_{2}}{e_{2} \app c}{\return{c_{1}}} }\right)}{\predA\left(r_{1}\right)}\\
        = & \after{r_{1}}{ \encode{0}{0 \bullet \left(e_{1} \bullet c\right)\bullet \left(e_{2} \bullet c\right)} \app \encode{1}{0}  }{\predA\left(r_{1}\right)}\\
        = & \after{r_{1}}{ \encode{0}{0 \bullet \left(e_{1} \bullet c\right)\bullet \left(e_{2} \bullet c\right)} \app \codepA  }{\predA\left(r_{1}\right)}
    \end{array}
    \]
    
    Similarly, from (2) we obtain:
    \[ \pred\left(c\right) \leq \after{r_{2}}{ \encode{0}{0 \bullet \left(e_{1} \bullet c\right)\bullet \left(e_{2} \bullet c\right)} \app \codepB  }{\predB\left(r_{2}\right)} \]

    Combining both, we get:
    \[
    \begin{array}{cl}
         & \pred\left(c\right)\\
         \leq & \hphantom{\meet \;} \after{r_{1}}{ \encode{0}{0 \bullet \left(e_{1} \bullet c\right)\bullet \left(e_{2} \bullet c\right)} \app \codepA  }{\predA\left(r_{1}\right)}\\
         &  
         \meet \; \after{r_{2}}{ \encode{0}{0 \bullet \left(e_{1} \bullet c\right)\bullet \left(e_{2} \bullet c\right)} \app \codepB  }{\predB\left(r_{2}\right)}\\
         = & \left(\predA \emeet \predB\right)\left(\encode{0}{0 \bullet \left(e_{1} \bullet c\right)\bullet \left(e_{2} \bullet c\right)}\right)\\
         \leql{\text{A.Return}} & \after{r}{\return{ \encode{0}{0 \bullet \left(e_{1} \bullet c\right)\bullet \left(e_{2} \bullet c\right)} }}{ \left(\predA \emeet \predB\right)\left(r\right) }\\
         = & \after{r}{\codepair{e_{1}}{e_{2}} \app c}{ \left(\predA \emeet \predB\right)\left(r\right) }
    \end{array}
    \]
    
    For \emph{elim1} we need to show:
     \[ \evexpand{\left(\predA \emeet \predB\right)}{\code{\codef{fst}}}{\predA} \]
    For $c \in \mca$:
    \[\begin{array}{cl}
         & \left(\predA \emeet \predB\right)\left(c\right)\\
        = &
          \after{c_{1}}{c \app \codepA}{ \predA\left(c_{1}\right) } \meet \after{c_{2}}{c \app \codepB}{ \predB\left(c_{2}\right) }
        \leq  \after{c_{1}}{c \app \codepA}{ \predA\left(c_{1}\right) }
         \\
        =&  \after{c_{1}}{ \encode{0}{0 \bullet \codepA} \app c }{ \predA\left(c_{1}\right) } 
        \\
        = & \after{c_{1}}{ \code{\codef{fst}} \app c }{ \predA\left(c_{1}\right) }
    \end{array}\]
    The proof of  \emph{elim2} is analogous. 

    \item[Universal Implication] Take:
    $$\begin{array}{r@{\hspace{0.05cm}}c@{\hspace{0.05cm}}l}
    \phi \imp \vec{\phi} \left(e\right)  &\defeq & \infimum_{\phi \in \vec{\phi}}\infimum_{c \in \mca}\left(\phi\left(c\right) \sqsupset \after{r}{e \app c}{\phi\left(r\right)}\right)
    \\ 
    \elambda{e} &\defeq & \encode{1}{e \bullet \left(\encode{2}{2 \bullet 0 \bullet 1} \bullet 0 \bullet 1\right)}\\ 
    \eeval &\defeq & \encode{0}{\eid \bullet \left( 0 \bullet \codepA\right) \bullet \left( 0 \bullet \codepB\right)}
    \end{array}$$
    For \emph{intro}, we need to show: 
    \[\dfrac{\forall \predC \in \preds \evexpand{\left(\predA \emeet \predB\right)}{e}{\predC}}{ \evexpand{\predA}{\codecurry{e}}{\left(\predB \imp \preds\right)} }\]
    Expanding the assumption, we get for all $\predC \in \preds$, and $c \in \mca$:
    \[ \after{c_{1}}{c \app \codepA}{\predA\left(c_{1}\right)} \meet \after{c_{2}}{c \app \codepB}{\predB\left(c_{2}\right)}
        \leq  \after{r}{e \app c}{\predC\left(r\right)} \]
    Instantiating with $c=\codetuple{c_{1}}{c_{2}}$, for which we have  
$ \codetuple{c_{1}}{c_{2}} \app \codepA = \return{ c_{1}}$ and  $\codetuple{c_{1}}{c_{2}} \app \codepB = \return{c_{2}}$, obtains: 
    \[ \after{c_{1}'}{\return{c_{1}}}{\predA\left(c_{1}'\right)}
        \meet 
        \after{c_{2}'}{\return{c_{2}}}{\predB\left(c_{2}'\right)} \leq  \after{r}{e \app \codetuple{c_{1}}{c_{2}}}{\predC\left(r\right)} \]
    From this, using after-ret, we get: 
    \[
        \predA\left(c_{1}\right) \meet \predB\left(c_{2}\right) \leq \after{r}{e \app \codetuple{c_{1}}{c_{2}}}{\predC\left(r\right)}
    \]
    By currying $\predB\left(c_{2}\right)$ we get:
    \[
        \predA\left(c_{1}\right) \leq \predB\left(c_{2}\right) \haimp \after{r}{e \app \codetuple{c_{1}}{c_{2}}}{\predC\left(r\right)}
    \]
    Since this holds for all $\predC \in \preds$ and $c_{2} \in \mca$, we get, by the infimum property, for all $c_{1} \in \mca$:
    \[
    \begin{array}{@{}c@{\hspace{-0.2cm}}l}
         & \predA\left(c_{1}\right) \\
         \leq & \infimum_{\predC \in \preds \, , \, c_{2} \in \mca}\left( \predB\left(c_{2}\right) \haimp \after{r}{e \app \codetuple{c_{1}}{c_{2}}}{\predC\left(r\right)}\right)\\
         = & \infimum_{\predC \in \preds \, , \, c_{2} \in \mca}\left( \predB\left(c_{2}\right) \haimp \after{r}{ \nu\left( e \bullet \left(\encode{2}{2 \bullet 0 \bullet 1} \bullet c_{1} \bullet c_{2} \right)\right) }{\predC\left(r\right)}\right)\\
         = & \infimum_{\predC \in \preds \, , \, c_{2} \in \mca}\left( \predB\left(c_{2}\right) \haimp \after{r}{ \encode{0}{ e \bullet \left(\encode{2}{2 \bullet 0 \bullet 1} \bullet c_{1} \bullet 0 \right) } \app c_{2} }{\predC\left(r\right)}\right)\\
         \leql{\text{A.Return}} & \after{f}{\return{\encode{0}{ e \bullet \left(\encode{2}{2 \bullet 0 \bullet 1} \bullet c_{1} \bullet 0 \right) }}}{ \infimum_{\predC \in \preds \, , \, c_{2} \in \mca}\left( \predB\left(c_{2}\right) \haimp \after{r}{ f \app c_{2} }{\predC\left(r\right)}\right) }\\
         = & \after{f}{ \codecurry{e} \app c_{1} }{\left( \predB \imp \preds \right)\left(f\right)}\\
    \end{array}
    \]
    
    For \emph{elim}, we need to show:
        \[ \forall \predB \in \preds \evrel{\left(\left(\predA \imp \preds\right) \emeet \predA\right)}{\codeeval}{\predB}  \]
    For this, we first prove the following uncurrying lemma, in which we use the abbreviation $\codeuncurry{e} \defeq \encode{0}{e \bullet \left(0 \bullet \codepA\right) \bullet \left(0 \bullet \codepB\right)}$.

\begin{lemma}[Universal Implication Uncurrying]\label{lemma:uni-imp-uncurry}
    \[ \dfrac{\evrel{\predA}{e}{\left(\predB \imp \preds\right)}}{\forall \predC \in \preds .\evrel{\left(\predA \emeet \predB\right)}{\codeuncurry{e}}{\predC}}  \]
\end{lemma}
\begin{proof}
    We need to prove:
    \[ \dfrac{\evexpand{\predA}{e}{\left(\predB \imp \preds\right)}}{\forall \predC \in \preds .\evexpand{\left(\predA \emeet \predB\right)}{\codeuncurry{e}}{\predC}}  \]
    Given a $\predC \in \preds$, due to the lower bound property of $\infimum$, we have:
    \[
    \begin{array}{cl}
         & \infimum_{\predC \in \preds}\infimum_{c_{2} \in \mca}\left(\predB\left(c_{2}\right) \haimp \after{r}{f \app c_{2}}{\predC\left(r\right)}\right) \\
         \leq & \infimum_{c_{2} \in \mca}\left(\predB\left(c_{2}\right) \haimp \after{r}{f \app c_{2}}{\predC\left(r\right)}\right)\\
         \leql{\eqref{eq:after-inter-mono}} & \after{c_{2}}{c \app \codepB}{\predB\left(c_{2}\right)} \haimp \after{c_{2}}{c \app \codepB}{\after{r}{f \app c_{2}}{\predC\left(r\right)}}\\
    \end{array}
    \]

    By \Cref{lemma:after-mono}, we get that for any $c_{1} \in \mca$:
    \[
    \begin{array}{cl}
         & \after{f}{e \app c_{1}}{\infimum_{\predC \in \preds}\infimum_{c_{2} \in \mca}\left(\predB\left(c_{2}\right) \haimp \after{r}{f \app c_{2}}{\predC\left(r\right)}\right)} \\
         \leq & \after{f}{e \app c_{1}}{\after{c_{2}}{c \app \codepB}{\predB\left(c_{2}\right)} \haimp \after{c_{2}}{c \app \codepB}{\after{r}{f \app c_{2}}{\predC\left(r\right)}}}
    \end{array}
    \]

    Instantiating our premise with $c_{1}$ obtains:
    \[ \begin{array}{cl}
         & \predA\left(c_{1}\right)\\
         \leq & \after{f}{e \app c_{1}}{\infimum_{\predC \in \preds}\infimum_{c_{2} \in \mca}\left(\predB\left(c_{2}\right) \haimp \after{r}{f \app c_{2}}{\predC\left(r\right)}\right)}\\
        \leql{\Cref{lemma:after-mono}} & \after{f}{e \app c_{1}}{\left(\after{c_{2}}{c \app \codepB}{\predB\left(c_{2}\right)} \haimp \after{c_{2}}{c \app \codepB}{\after{r}{f \app c_{2}}{\predC\left(r\right)}}\right)}\\
        \leql{\Cref{lemma:after-imp}} & \after{c_{2}}{c \app \codepB}{\predB\left(c_{2}\right)} \haimp \after{f}{e \app c_{1}}{ \after{c_{2}}{c \app \codepB}{\after{r}{f \app c_{2}}{\predC\left(r\right)}}}
    \end{array} \]

    So by uncurrying $\after{c_{2}}{c\app\codepB}{\predB\left(c_{2}\right)}$ and then currying $\predA\left(c_{1}\right)$ we get:
    \[
    \begin{array}{cl}
        & \after{c_{2}}{c\app\codepB}{\predB\left(c_{2}\right)}\\
        \leq & \predA\left(c_{1}\right) \haimp \after{f}{e \app c_{1}}{ \after{c_{2}}{c \app \codepB}{\after{r}{f \app c_{2}}{\predC\left(r\right)}}}
    \end{array}
    \]

    Now, for every $c_{1} \in \mca$, we have that:
    \[
    \begin{array}{cl}
        & \after{c_{2}}{c\app\codepB}{\predB\left(c_{2}\right)}\\
        \leq &\infimum_{c_{1}} \left( \predA\left(c_{1}\right) \haimp \after{f}{e \app c_{1}}{ \after{c_{2}}{c \app \codepB}{\after{r}{f \app c_{2}}{\predC\left(r\right)}}}\right)\\
        \leql{\eqref{eq:after-inter-mono}} & \after{c_{1}}{c \app \codepA}{\predA\left(c_{1}\right)} \haimp \after{c_{1}}{c \app \codepA}{\after{f}{e \app c_{1}}{ \after{c_{2}}{c \app \codepB}{\after{r}{f \app c_{2}}{\predC\left(r\right)}}}}
    \end{array}
    \]

    Therefore,  by uncurrying again we get: 
    \[    
    \begin{array}{cl}
         & \left(\predA \emeet \predB\right)\left(c\right)\\
         = &\after{c_{1}}{c \app \codepA}{\predA\left(c_{1}\right)} \meet \after{c_{2}}{c \app \codepB}{\predB\left(c_{2}\right)}\\
         \leq & \after{c_{1}}{c \app \codepA}{\after{f}{e \app c_{1}}{\after{c_{2}}{c \app \codepB}{\after{r}{f \app c_{2}}{\predC\left(r\right)}}}} \\
         \leql{\text{A.Bind}} & \after{r}{\left( \letin{c_{1}}{c \app \codepA}{\letin{f}{e \app c_{1}{\letin{c_{2}}{c \app \codepB}{f \app c_{2}}}}} \right)}{\predC\left(r\right)} \\
         = & \after{r}{ \nu\left(e \bullet \left(c \bullet \codepA\right) \bullet \left(c \bullet \codepB\right)\right)}{\predC\left(r\right)}\\
         = & \after{r}{\encode{0}{e \bullet \left(0 \bullet \codepA\right) \bullet \left(0 \bullet \codepB\right)} \app c}{\predC\left(r\right)} \\
         = & \after{r}{\codeuncurry{e} \app c}{\predC\left(r\right)}
    \end{array}
    \]
\end{proof}

    To get back to the proof of \emph{elim}, by Reflexivity, we have:
    \[{\evrel{\predA \imp \preds}{\codeid}{\predA \imp \preds}} \]
    From which, by~\Cref{lemma:uni-imp-uncurry}, we get:
    \[ {\forall \predB \in \preds .\evrel{\left(\left(\predA \imp \preds\right) \emeet \predA\right)}{\codeuncurry{\codeid}}{\predB}} \]
    From this, we get the desired result, since $\codeeval = \codeuncurry{\codeid}$.  
\qedhere
\end{description}
\end{proof}

\sepconsistent*
\begin{proof}
    ($\Rightarrow$) If there exists $e\in\separator$ such that $\top \evidence{e} \bot$, then for every $c \in \mca$:
    \[ \boldsymbol{1} = \top\left(c\right) \leq \after{r}{e \app c}{\bot\left(r\right)} = \after{r}{e \app c}{\boldsymbol{0}} \leq \boldsymbol{0} \]

    ($\Leftarrow$) If $\boldsymbol{1} \leq \boldsymbol{0}$, then $\top \evidence{\codeid} \bot $ since for every $c \in \mca$:
    \[ \top\left(c\right) = \boldsymbol{1} \leq \boldsymbol{0} \leql{\text{A.Return}} \after{r}{\return{c}}{\boldsymbol{0}} = \after{r}{\codeid \app c}{\bot\left(r\right)} \qedhere\]
\end{proof}

}
\end{document}